\documentclass[12pt]{article}
\usepackage{setspace}
\usepackage{authblk}
\usepackage{booktabs}
\usepackage{natbib}
\usepackage[short,nodayofweek,12hr]{datetime}
\usepackage{float}
\usepackage{xcolor}
\usepackage{amsfonts}
\usepackage{amsmath}
\usepackage{amsthm}
\usepackage{chngcntr}
\usepackage{bbm}
\usepackage{adjustbox}
\usepackage[flushleft]{threeparttable}
\usepackage{graphics}
\usepackage{hyperref}
\usepackage{booktabs}
\usepackage{graphicx}
\usepackage{comment}
\usepackage[utf8]{inputenc}
\usepackage{sectsty}
\usepackage[textsize=footnotesize]{todonotes}
\usepackage{lipsum}
\usepackage{abstract}
\usepackage[margin=1.25in]{geometry}

\usdate

\definecolor{myblue}{RGB}{0 0 128}
\hypersetup{
 colorlinks=true,
 urlbordercolor={1 1 1},
 urlcolor={blue},
 linkcolor={red},
 citecolor={myblue},
 raiselinks=false
}

\newcommand{\boldr}{\boldsymbol{r}}
\newcommand{\bolda}{\boldsymbol{a}}
\newcommand{\boldw}{\boldsymbol{w}}
\newcommand{\boldh}{\boldsymbol{h}}
\newcommand{\boldl}{\boldsymbol{\ell}}
\newcommand{\boldz}{\boldsymbol{z}}
\newcommand{\boldW}{\boldsymbol{W}}

\newcommand{\boldc}{\boldsymbol{c}}
\newcommand{\boldsigma}{\boldsymbol{\sigma}}
\newcommand{\boldzeta}{\boldsymbol{\zeta}}

\newcommand{\boldLambda}{\boldsymbol{\Lambda}}

\newcommand{\boldrho}{\boldsymbol{\rho}}
\newcommand{\boldvarrho}{\boldsymbol{\varrho}}

\newcommand{\vech}{\operatorname{vech}}

\newtheorem{proposition}{Proposition}

\newtheorem{assumption}{Assumption}

\newtheorem*{remark}{Remark}

\title{\textbf{Spot Regressions with Candlesticks}\footnote{This paper is built on the third chapter of my PhD dissertation at Duke University. I am very grateful to Andrew J. Patton, Tim Bollerslev and Anna Bykhovskaya for their guidance and support. I would also like to thank Anna Cieslak, Mehmet Caner, Peter R. Hansen, Ilze Kalnina, Raul Guarini Riva, Adam Rosen, George Tauchen, Christopher Walker and seminar participants at Duke University, 2025 Triangle Econometrics Conference, and 2025 Annual Meeting of SoFiE for their helpful comments and suggestions. Email: \textcolor{blue}{\href{mailto:yasin.simsek@duke.edu}{yasin.simsek@duke.edu}}.}}
\author{\vspace{0.5cm}Yasin Simsek}
\affil{\textit{\small Department of Economics, Duke University}}
\date{\today}
\begin{document}

%% TITLE PAGE
\maketitle

% Set line spacing
\doublespacing

%% ABSTRACT
\begin{abstract}
\noindent 
Betas from spot regressions are central to asset pricing and risk management, as measures of systematic risk. This paper develops a new estimation and inference framework for spot regressions by leveraging high-frequency candlesticks, extending conventional (open-to-close) returns with intra-period high/low prices. Specifically, I construct candlestick-based estimators of regression parameters, including spot beta, by minimizing a quadratic risk under a fixed-k asymptotic framework. I then develop a feasible hypothesis testing procedure for spot betas with correct asymptotic size. Simulation results show that the proposed estimator reduces estimation risk relative to return-based estimators, especially in small samples, and the test achieves notably higher power. I apply the framework to assess the market neutrality of Bitcoin using 1-minute data on IBIT and SPY, finding deviations from neutrality, particularly in high-volatility periods.
\medskip


\end{abstract}

\newpage
\section{Introduction}
Time-series return regressions are widely employed in empirical finance to measure an asset's exposure to risk factors, most prominently the market portfolio. The slope coefficients of these regressions are commonly referred to as betas and they are central to asset pricing and risk management research. The conventional approach for running these regressions relies on low-frequency returns (e.g., daily or weekly) over long horizons (e.g., months or years), thereby treating betas as constants throughout the estimation window. However, a substantial body of evidence shows that risk exposures are time-varying; see for example \cite{engle2004risk}. High-frequency intraday returns provide an effective means for capturing this variation and have enabled more precise estimation over shorter horizons, such as a day or a month; see for example \cite{ait2014high} and references therein\footnote{Modeling time variation in betas is not limited to high-frequency methods and researchers have developed various approaches over time. One strand of the literature, for example, captures the time variation by modelling betas as parametric functions of conditioning variables (e.g., \cite{connor2012efficient,gagliardini2016time}) while another strand adopts nonparametric methods (e.g., \cite{ang2012testing}).}. More recently, a growing literature documents that betas vary significantly even within a single trading day (e.g., \cite{bibinger2019estimating, andersen2021recalcitrant, andersen2023intraday, liao2024changes,patton2024intraday}). These findings motivate the estimation of risk exposures over very short intraday windows, commonly known as spot regressions. Since such spot estimates are constructed from returns observed over narrow time intervals, they are subject to limited local information, giving rise to a fundamental bias-variance trade-off between statistical precision and localization. 

%Sampling returns at finer frequencies is a natural way to address this trade-off, but doing so introduces market microstructure noise, requiring additional modeling assumptions to form an estimator.

Building on this background, I develop a novel estimation framework for spot regressions using a new source of information to mitigate the localization-precision trade-off. Specifically, my approach leverages the richer information in high-frequency candlesticks that contain the open, high, low, and close prices within each sampling interval, moving beyond the conventional methods that rely solely on (open-to-close) returns. Importantly, this work provides a practical way to obtain more precise estimates and reliable inference, while preserving the local structure of the estimation.

I first construct a spot covariance estimator that exploits the second moments of candlestick variables and then derive estimates for regression parameters, including spot beta, implied by this estimator. The functional form of this estimator is obtained by minimizing a well-defined quadratic risk function. The resulting expression is analytically tractable and resembles the least squares formula. These features render the estimator readily applicable across settings without requiring additional econometric procedures. Importantly, my approach treats the size of the estimation window (equivalently, the number of observations in the estimation sample) as a fixed, possibly small, constant, thereby accounting for the scarcity of local information. I further show that the proposed estimator substantially reduces risk relative to standard estimators that rely solely on return observations, particularly in small samples. These properties make the approach especially well suited to high-frequency event studies, which require highly local estimates for identification (e.g., \cite{bollerslev2018modeling}; \cite{nakamura2018high}).\footnote{Specifically, \cite{nakamura2018high} exploit high-frequency bond returns within a short window around FOMC announcements to identify the effects of monetary policy shocks, while \cite{bollerslev2018modeling} estimate investor disagreement using local jump regressions around news announcements.}

%Specifically, I first construct a spot covariance estimator by exploiting the second moments of candlestick variables. Then, I consider spot beta estimates implied by this covariance estimator. The form of the estimator is determined by minimizing a well-defined quadratic risk function, thereby improving the efficiency of the estimator relative to existing methods. The resulting form is analytically tractable, resembling the least squares formula and does not depend on the specific sample at hand. This feature makes the estimator readily applicable to various settings without requiring further econometric procedures. Importantly, my approach considers the size of the estimation window (or number of observations in the estimation sample) as a fixed, possibly small,constant, which effectively accounts for the scarcity of local information. Moreover, I find that my estimator substantially reduces the risk relative to the standard estimators that only use return observations, especially in small samples. This makes my approach particularly useful for high-frequency event studies which require highly local estimates for identification (e.g. \cite{bollerslev2018modeling}, \cite{nakamura2018high}).\footnote{Specifically, \cite{nakamura2018high} exploit high-frequency bond returns over a short window around FOMC announcements to identify the effects of monetary policy shocks. \cite{bollerslev2018modeling} estimate disagreement among investors by local jump regressions around news announcements.}
 
Having established the candlestick-based estimator, I further develop a formal inference procedure on spot betas using the proposed estimator. While recent studies have investigated candlestick-based inference on spot volatility functionals (e.g., \cite{li2024reading,bollerslev2024optimal}), inference methods for spot betas remain underexplored. To address this gap, I propose a hypothesis test and associated test statistic. I show that this test statistic can be approximated by a limiting variable whose distribution under the null hypothesis can be characterized via simulations. Accordingly, I compute the critical values by simulating the quantiles of that limiting variable. The resulting procedure delivers a test with asymptotically correct size. A finite-sample simulation study further demonstrates that the test exhibits reasonable size control and achieves substantially higher power than its return-based counterpart (\cite{bollerslev2024optimalspotregressions}).

Finally, I demonstrate the practical value of the proposed framework through an empirical application studying the market neutrality of Bitcoin. Cryptocurrencies are often framed as ``digital gold'' by their advocates, suggesting potential hedging properties. Similarly, \cite{liu2021risks} find limited evidence of systematic exposure of crypto assets to market risk. Motivated by these observations, I test the null hypothesis of market neutrality (or zero market beta) for Bitcoin using my candlestick-based estimator and test. Specifically, I use high-frequency candlestick observations of the iShares Bitcoin ETF (IBIT) and the SPDR S\&P 500 ETF (SPY). The analysis uncovers striking patterns with potential implications for risk management. I reject the null approximately 35\% of the time, revealing pricing dynamics that differ markedly from those documented in prior work, such as \cite{liu2021risks}. Notably, rejection rates peak around August and September 2024, coinciding with elevated financial market volatility. This is precisely when the diversification and hedging benefits of Bitcoin would be most valuable.

In an attempt to improve the estimation, one could alternatively employ tick-level data which provides the most granular information, as it records every single transaction in the market usually at ultra-high frequencies (e.g., milliseconds). However, using tick-level data faces important limitations. First, this data is accessible only to well-resourced researchers since it requires costly subscriptions through commercial providers such as NYSE Trade and Quotes (TAQ) or TickData. Furthermore, prices sampled at such ultra-high frequencies are invariably contaminated by market microstructure noise, which necessitates imposing additional modeling assumptions on the noise structure to obtain enhanced estimates, see for example \cite{diebold2013correlation}. By contrast, candlestick data is widely accessible through many different public sources at ``not-too-fine'' frequencies (e.g., 1-minute or 5-minute), which naturally guards against the impact of microstructure noise. From this perspective, my paper adopts a more practical and accessible approach for improving spot regression estimation.\footnote{Within the high-frequency econometrics literature, researchers have developed models for market microstructure noise; see, for example, \cite{hayashi2005covariance, ait2010high, christensen2010pre, barndorff2011multivariate}. While such approaches may be relevant in this context, their incorporation into the present framework is not straightforward. I leave a careful investigation of these issues for future research.}

This paper contributes to multiple strands of the literature. From a technical perspective, this paper is closely related to the high-frequency econometrics literature on beta estimation, see for example \cite{barndorff2004econometric,mykland2009inference,li2017adaptive,ait2020high,bollerslev2024optimalspotregressions} among many others. These papers develop estimators primarily based on high-frequency returns from open and close prices. My work complements this literature by introducing a new class of estimators that leverage candlestick observations, which expands the information set of high-frequency intervals with high and low prices. 

My work is inspired by the range-based volatility estimation literature. Starting from seminal papers by \cite{garman1980estimation} and \cite{parkinson1980extreme}, this literature highlights remarkable efficiency gains in estimating variances by extracting more information from candlestick prices. Notably, \cite{christensen2007realized} introduces the realized-range estimator for integrated variance, constructed by high-frequency ranges. Later, it is extended to be robust to microstructure noise (\cite{martens2007measuring, christensen2009bias}); jumps (\cite{christensen2012asymptotic}) and drifts (\cite{li2025realized}). More recently, several papers study spot volatility estimation using candlesticks, see for example \cite{li2024reading}, \cite{bollerslev2024optimal} and \cite{bollerslev2024newtricks}. The main focus of these papers is to derive optimal estimation and inference frameworks for spot volatility and its functionals. This paper complements this literature by extending these ideas to multivariate settings which allows us to study dependence structure between assets with this data.

This paper is perhaps most closely related to \cite{bollerslev2024optimal} which develops a decision-theoretic framework to construct optimal estimators for volatility functionals. In their approach, the risk function is asymptotically approximated using pivotal random variables, and the asymptotic risk is minimized through Monte Carlo simulations. However, in multivariate settings, the limiting distributions are generally non-pivotal, which makes direct minimization of a traditional risk function infeasible. To overcome this challenge, I integrate the asymptotic risk to obtain an average risk function which is pivotal, allowing me to determine the optimal estimators by minimizing this objective. Importantly, this approach delivers a feasible solution to an otherwise intractable problem while delivering superior efficiency in both estimation and inference. 

Prior to this study, \cite{brandt2006no} and \cite{bannouh2009range} introduced range-based covariance estimators that exploit triangular no-arbitrage conditions in currency markets. While effective in that setting, this approach is generally not applicable to equity markets. Similarly, \cite{rogers2008estimating} estimated the correlation coefficient of a bivariate Brownian motion using open, high, low and close price observations. However, their analysis is conducted under a constant volatility framework. In contrast, this paper accommodates a more general Itô semimartingale environment and provides a feasible inference procedure with candlesticks which is not considered in the abovementioned papers.

The remainder of the paper is organized as follows. Section \ref{sec:setup} introduces the theoretical environment and notation. Section \ref{sec:estimation} develops candlestick-based estimation framework. Section \ref{sec:inference} discusses the inference procedure. Section \ref{sec:simulation} presents simulation results, and Section \ref{sec:empirics} provides the empirical application. Finally, Section \ref{sec:conclusion} concludes. Additional technical details and all proofs are provided in the appendices.

%The adverse effects of a poorly estimated covariance matrix have been widely recognized by the portfolio selection literature. As \cite{li2015sparse} and the references therein show, the estimation error in covariances often leads to highly leveraged positions on certain assets\footnote{\cite{kelly2024universal} argue that efficient portfolios take large bets on low-variance principal component portfolios. This leads to poor out-of-sample performance since these portfolios have small risk premia.}, consequently diverging from ``true'' optimal weights and offering low reward-risk ratio in the out-of-sample. Indeed, our approach may be seen as a way to address these issues by integrating candlestick data into the estimation process, ultimately enhancing the information content of the training data and yielding higher-quality covariance estimates. To elucidate this idea, we conduct an empirical analysis in estimating global minimum variance (GMV) portfolios. We study the high-frequency price data of S\&P 500 stocks between 1996 and 2023. Using the integrated covariance estimates based on the MAR estimator of spot covariance, we construct minimum-variance portfolios with various sizes and compare their out-of-sample performance with standard benchmarks like realized covariance estimator and equally-weighted portfolio. Empirical results show that our approach achieves significantly reduced portfolio risk in the out-of-sample while maintaining lower turnover, smaller aggregate short positions, and less extreme exposures to individual assets.

\section{Setup and Notation} \label{sec:setup}
%\textcolor{red}{REWRITE}Section \ref{sec:price_process} introduces the underlying price process. Section \ref{sec:obs_scheme} presents the observation scheme and defines the candlestick returns. Section \ref{sec:estimator} proposes a class of new estimators for spot covariance. Then, Section \ref{sec:optimal_weight} solves for the optimal estimator in the proposed class. Finally, Section \ref{sec:discussion} discusses the implications.

\subsection{The Price Process} \label{sec:price_process}
Consider the bivariate log-price process $\boldsymbol{X}_t = [X_{1,t},X_{2,t}]^{\top}$, observed at time $t$. Suppose that $\boldsymbol{X}_t$ evolves as an It\^o semimartingale on the filtered probability space $(\Omega, \mathcal{F}, (\mathcal{F}_t), \mathbb{P})$, according to
\begin{equation} \label{eq:itosemi}
    d\boldsymbol{X}_t = \boldsymbol{b}_t dt + \boldsymbol{\sigma}_t d\boldsymbol{W}_t,
\end{equation}
where $\boldsymbol{b}_t$ is the drift process, $\boldsymbol{\sigma}_t$ is stochastic volatility matrix and $\boldsymbol{W}_t\equiv [W_{1,t}, W_{2,t}]^{\top}$ is a bivariate standard Brownian motion. 

The spot covariance matrix, $\boldsymbol{c}_t$, and associated correlation matrix, $\boldsymbol{\rho}_t$, at time $t$ are defined as:
\begin{equation}
    \boldsymbol{c}_t \equiv \boldsymbol{\sigma}_t\boldsymbol{\sigma}_t^{\top} = \begin{pmatrix}
        c_{11,t} & c_{12,t} \\
        c_{21,t} & c_{22,t}
    \end{pmatrix} \quad \text{and} \quad
    \boldsymbol{\rho}_t \equiv \begin{pmatrix}
        1 & \rho_t \\
        \rho_t & 1
    \end{pmatrix},
\end{equation}
where $\rho_t = c_{12,t}/\sqrt{c_{11,t}c_{22,t}}$ is the spot correlation between $X_{1,t}$ and $X_{2,t}$.

The It\^o-semimartingale representation in Equation \eqref{eq:itosemi} can be motivated by no-arbitrage conditions and therefore serves as a workhorse framework in the continuous-time finance literature, see for example \cite{back2010asset}. Consequently, it has become the fundamental model for analyzing high-frequency asset prices; see \cite{jacod2012discretization} and \cite{ait2014high} for further discussions. For simplicity, I exclude jumps from the price process.\footnote{Standard jump-robust techniques such as truncation (\cite{mancini2009non}) or bi-power variation (\cite{barndorff2004power}) methods can in principle be adapted to this paper to explicitly account for price discontinuities.}

I further tailor the price process $\boldsymbol{X}_t$ to the following regression representation:
\begin{equation} \label{eq:reg_rep}
\begin{array}{rcl} 
    d {X}_{1,t} &=&  b_{1,t} + \nu_t^{1/2} dW_{1,t} \\
    d {X}_{2,t} &=&  b_{2,t} + \beta_t d {X}_{1,t} + \varsigma^{1/2}_t dW_{2,t}.
\end{array}
\end{equation}
This is also equivalent to imposing the following spot covariance structure:
\begin{equation} \label{eq:regression_covariance}
    \boldc_t = \begin{pmatrix}
        \nu_{t} & \beta_t \nu_{t} \\
        \beta_t \nu_{t}  & \beta_t^2 \nu_{t} + \varsigma_t
    \end{pmatrix} \quad \text{and} \quad
    \boldsymbol{\sigma}_t = \begin{pmatrix}
        \nu_t^{1/2} & 0 \\
        \beta_t \nu_t^{1/2} & \varsigma_t^{1/2}
    \end{pmatrix}.
\end{equation}
Through the lens of factor models in asset pricing literature (e.g., \cite{sharpe1964capital, lintner1965security,fama1992cross}), one can consider the first asset as the market portfolio and the second asset as a risky asset. Then, $\nu_t$ and $\varsigma_t$ refer to the market and idiosyncratic variances, respectively. Finally, $\beta_t$ denotes the market beta of the risky asset which is the main object of interest in this paper. 

The above representation implies that the spot beta, idiosyncratic variance and market variance can be identified directly from the spot covariance matrix $\boldc_t$ as follows:
\begin{equation} \label{eq:beta_relation}
    \beta_t  =  \frac{c_{12,t}}{c_{11,t}}, \quad
    \varsigma_t = c_{22,t} - \frac{c_{12,t}^2}{c_{11,t}}, \quad \text{and} \quad
    \nu_t  =  c_{11,t}.
\end{equation}
Throughout the paper, $(\beta_t, \varsigma_t, \nu_t)$ are referred to as the spot regression parameters.

\subsection{Observation Scheme and Candlestick Returns} \label{sec:obs_scheme}
The price process $\boldsymbol{X}_t$ is assumed to be sampled on a regular time grid, $\{i\Delta_n: i=0,1, ..., n \}$ over a fixed time span $[0,T]$. Here, $\Delta_n=T/n$ refers to the sampling frequency and $n$ is the number of observations, an integer. High-frequency intervals are denoted by $\mathcal{T}_{i} \equiv [(i-1)\Delta_n, i\Delta_n]$ for each $i\in \{1, ..., n\}$. Following the standard practice in the high-frequency financial econometrics literature, I consider an in-fill asymptotic framework where $\Delta_n \to 0$ asymptotically, see for example \cite{ait2014high}.\footnote{The choice of $\Delta_n$ is usually guided by volatility signature plots introduced in \cite{andersen2000great}. Following the standard practice in the high-frequency econometrics literature, I adopt moderate sampling frequencies like $\Delta_n=1,5,10$-min in my practical implementations, thereby avoiding ultra high-frequency observations that could be contaminated by microstructure effects.}

A typical candlestick over the interval $\mathcal{T}_{i}$ consists of four prices:
\begin{equation*}
    \begin{array}{lccr}
        \boldsymbol{X}_{(i-1)\Delta_n}, &
        \displaystyle\sup_{t \in \mathcal{T}_i} \boldsymbol{X}_t, &
        \displaystyle\inf_{t \in \mathcal{T}_i} \boldsymbol{X}_t, &
        \boldsymbol{X}_{i\Delta_n}
    \end{array}
\end{equation*}
which are called the open, high, low, and close prices, respectively. From these prices, one can construct the following normalized returns:
\begin{equation} \label{eq:candlestick_returns}
    \begin{array}{rcl}
        \boldsymbol{r}_i & \equiv & \frac{\boldsymbol{X}_{i\Delta_n}-\boldsymbol{X}_{(i-1) \Delta_n}}{\sqrt{\Delta_n}}, \\[1em]
        
        \boldsymbol{h}_i & \equiv & \frac{\sup _{t \in \mathcal{T}_i} \boldsymbol{X}_t-\boldsymbol{X}_{(i-1) \Delta_n}}{\sqrt{\Delta_n}}, \\[1em]
        
        \boldsymbol{\ell}_i & \equiv & \frac{\inf _{t \in \mathcal{T}_i} \boldsymbol{X}_t-\boldsymbol{X}_{(i-1) \Delta_n}}{\sqrt{\Delta_n}}
    \end{array}
\end{equation}
where $\sup$ and $\inf$ operators are applied element-wise. The first line defines the standard high-frequency (open-to-close) return, commonly employed in the literature. The second and third lines indicate the high-open and low-open returns, respectively. These returns stand as a new source of information in this framework. Indeed, the bundle $(\boldsymbol{r}_i, \boldsymbol{h}_i, \boldsymbol{\ell}_i)$ summarizes the price dynamics within the interval $\mathcal{T}_{i}$ through the lens of candlesticks. 

I also define two variables:
\begin{equation} \label{eq:candlestick_range}
    \begin{array}{rcl}
        \boldsymbol{w}_i & \equiv & \boldsymbol{h}_i - \boldsymbol{\ell}_i, \\
        \boldsymbol{a}_i & \equiv & \boldh_i + \boldl_i - \boldr_i.
    \end{array}
\end{equation}
where $\boldsymbol{w}_i$ and $\boldsymbol{a}_i$ refer to the range and asymmetry of the candlestick returns, respectively. One can show that $(\boldr_i, \bolda_i, \boldw_i)$ is an invertible linear transformation of $(\boldsymbol{r}_i, \boldsymbol{h}_i, \boldsymbol{\ell}_i)$, implying that both vectors carry the same information. Recent papers on candlestick-based volatility estimation (e.g., \cite{li2024reading}, \cite{bollerslev2024optimal} and \cite{bollerslev2024newtricks}) work with $(\boldr_i, \bolda_i, \boldw_i)$. Thus, to be consistent with this literature and facilitate the analysis, I also work with $(\boldr_i, \bolda_i, \boldw_i)$ for the rest of the paper. However, the proposed methodology can be easily adapted to the original set of returns $(\boldsymbol{r}_i, \boldsymbol{h}_i, \boldsymbol{\ell}_i)$ without changing any of the main results.

To build intuition, I present the graphical representation of the candlestick returns on a typical candlestick chart in Figure \ref{fig:candlestick_chart}. As shown in this figure, the range $\boldw_i$ is shown as the vertical distance between the high and low prices, while the return $\boldr_i$ reflects the length of the thick body. The asymmetry $\bolda_i$ indicates the position of the returns (or the thick body) within that range. For example, if $\bolda_i=0$ then the thick body of the candlestick is exactly centered between the high and low prices. On the other hand, if $\bolda_i>0$ then the body is skewed towards the low price, and vice versa. 

\begin{figure}
    \centering
    \includegraphics[width=1\textwidth]{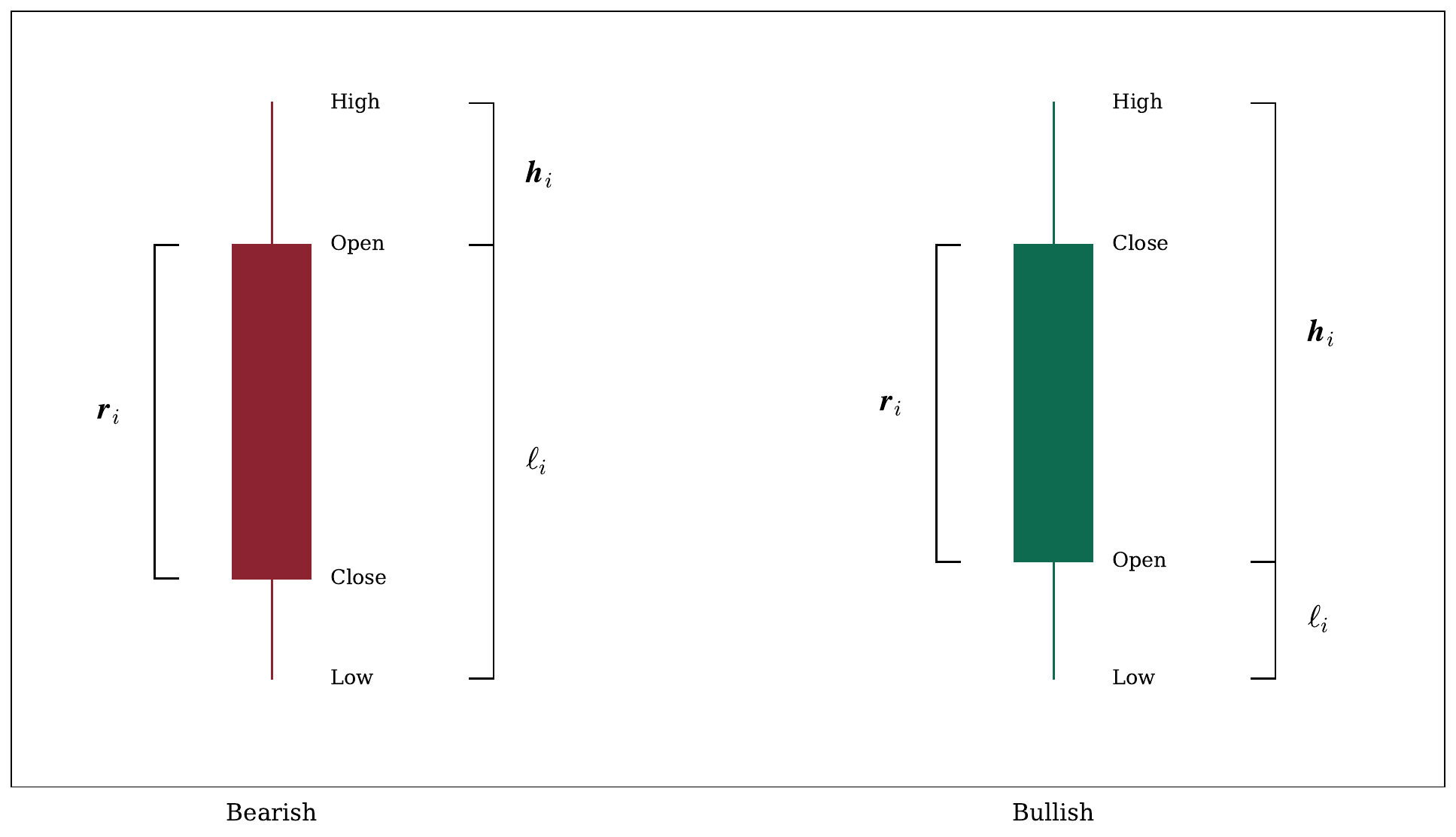}
    \caption{\textbf{Illustration of candlestick returns:} The figure shows examples of bearish (left) and bullish (right) candlesticks, illustrating the open, high, low, and close prices. Also it highlights the open-close return $\boldsymbol{r}_i$, high return $\boldsymbol{h}_i$, and low return $\boldsymbol{\ell}_i$ as defined in Equation \eqref{eq:candlestick_returns}.}
    \label{fig:candlestick_chart}
\end{figure}

\section{Candlestick-based Estimation} \label{sec:estimation}
The primary objective of this paper is to construct estimators for spot regression parameters that efficiently exploit the information contained in candlestick features. To this end, I first show how to construct spot covariance estimators using candlestick returns, and then obtain spot regression estimates exploiting the relation in Equation \eqref{eq:beta_relation}.

\subsection{A Class of Spot Covariance Estimators} 
A standard estimator for the spot covariance matrix $\boldc_t$ is given by the local average of the outer products of high-frequency (open-to-close) returns. For any $t \in [0, T]$, this can be formally expressed as:
\begin{equation} \label{eq:spotcov_estimator_old}
    \widehat{\boldc}_t = \frac{1}{k}\sum_{i\in \mathcal{I}_{n,t}} \boldsymbol{r}_{i} \boldsymbol{r}_{i}^{\top}.
\end{equation}
where $\mathcal{I}_{n,t} = \{\lceil \frac{t}{\Delta_n} \rceil + 1, \ldots, \lceil \frac{t}{\Delta_n} \rceil + k\}$ is the local estimation window and $k$ stands for the size of that window.\footnote{Here, $\lceil \cdot \rceil$ denotes the ceiling function, which maps a real number to the smallest following integer. Thus, $\mathcal{I}_{n,t}$ can be interpreted as the right-sided local window that contains $k$ observations after time $t$. Alternatively, left-sided or symmetric windows can be considered without changing the main results of the paper.} Correspondingly, the spot quantities in the regression representation in Equation \eqref{eq:reg_rep} can be estimated as follows:
\begin{equation} \label{eq:spotcov_estimator_old_beta}
    \widehat{\beta}_{n,t} = \frac{\widehat{c}_{12,t}}{\widehat{c}_{11,t}}, \quad
    \widehat{\varsigma}_{n,t} = \widehat{c}_{22,t} - \frac{\widehat{c}_{12,t}^2}{\widehat{c}_{11,t}}, \quad \text{and} \quad
    \widehat{\nu}_{n,t} = \widehat{c}_{11,t}.
\end{equation}

The estimator in Equation \eqref{eq:spotcov_estimator_old} is well-studied in the high-frequency econometrics literature; see, among others, \cite{fan2008spot}. While the conventional estimators let $k=k_n$ grow with sample size $n$, analogous to bandwidth parameter in classic kernel nonparametrics, I treat $k$ as a fixed constant, not dependent on $n$. Accordingly, throughout this paper I adopt fixed-$k$ asymptotics developed by \cite{bollerslev2021fixed}. This approach was originally proposed for inference on spot volatilities and subsequently extended to spot beta inference by \cite{bollerslev2024optimalspotregressions}. The framework yields inference that remains valid for a fixed, potentially small, value of \(k\), thereby directly capturing the local nature of spot estimates rather than relying on asymptotic approximations that require \(k\) to diverge.\footnote{Fixed-$k$ asymptotics build on coupling (or strong approximation) arguments; see \cite{jacod2021volatility}. The key idea is that, over short intervals, an Itô semimartingale can be locally approximated by a scaled Brownian motion. This feature allows for determining the (approximate) finite-sample distribution of the estimator under shrinking $\Delta_n$ (or diverging $n$). \cite{bollerslev2021fixed} show that fixed-$k$ asymptotics deliver confidence intervals with more accurate coverage than conventional large-$k$ approximations, especially in small samples, where the latter often suffer from nontrivial size distortions.}

By construction, Equation \eqref{eq:spotcov_estimator_old} relies solely on open-to-close returns, discarding potentially valuable information embedded in other candlestick returns. Thus, I introduce a flexible class of estimators for $\boldsymbol{c}_t$ that incorporates all candlestick returns $(\boldsymbol{r}_i, \boldsymbol{a}_i, \boldsymbol{w}_i)$ and combines them through a weighted sum of quadratic forms over a local window. 

Formally, the estimator is defined as:
\begin{equation} \label{eq:spotcov_estimator}
    \begin{aligned}
        \widehat{\boldsymbol{c}}_{n, t}(\lambda) &= \frac{1}{k} \sum_{i\in \mathcal{I}_{n,t}} \Big\{\lambda_{1} \boldr_{i} \boldr_{i}^{\top} + \lambda_{2} \bolda_{i} \bolda_{i}^{\top} + \lambda_{3} \boldw_{i} \boldw_{i}^{\top} \\
         &+ \lambda_{4}(\boldr_{i} \bolda_{i}^{\top} + \bolda_{i} \boldr_{i}^{\top}) + \lambda_{5}(\boldr_{i} \boldw_{i}^{\top} + \boldw_{i} \boldr_{i}^{\top}) + \lambda_{6}(\bolda_{i} \boldw_{i}^{\top} + \boldw_{i} \bolda_{i}^{\top}) \Big\}
    \end{aligned}
\end{equation}
where $\lambda \equiv (\lambda_{1}, \lambda_{2}, \lambda_{3}, \lambda_{4}, \lambda_{5}, \lambda_{6}) \in \mathbb{R}^6$ denote the weights assigned to each   component. This estimator can be compactly expressed in matrix form as follows:
\begin{equation}
    \widehat{\boldsymbol{c}}_{n, t}(\lambda) = \frac{1}{k} \sum_{i\in \mathcal{I}_{n,t}} \boldz_i \Lambda \boldz_i^{\top}, \quad \text{where} \quad \boldz_i = \begin{pmatrix}
        \boldr_i &
        \bolda_i &
        \boldw_i
    \end{pmatrix} \quad \text{and} \quad\Lambda = \begin{pmatrix}
        \lambda_1 & \lambda_4 & \lambda_5 \\
        \lambda_4 & \lambda_2 & \lambda_6 \\
        \lambda_5 & \lambda_6 & \lambda_3
    \end{pmatrix}.
\end{equation}
The resulting estimator is always symmetric and positive semi-definite whenever $\Lambda$ is positive semi-definite. The choice of \(\lambda\) will be discussed in subsequent sections in detail. Naturally, I propose the following estimates for the spot quantites in the regression representation:

\begin{equation} \label{eq:spotbeta_estimator}
    \widehat{\beta}_{n,t}(\lambda) = \frac{\widehat{c}_{12,t}(\lambda)}{\widehat{c}_{11,t}(\lambda)}, \quad
    \widehat{\varsigma}_{n,t}(\lambda) = \widehat{c}_{22,t}(\lambda) - \frac{\widehat{c}_{12,t}^2(\lambda)}{\widehat{c}_{11,t}(\lambda)}, \quad \text{and} \quad
    \widehat{\nu}_{n,t}(\lambda) = \widehat{c}_{11,t}(\lambda).
\end{equation}

The proposed class of estimators is highly flexible. For example, setting $\lambda=(1,0,0)$ recovers the classic spot covariance and beta estimator defined in Equation \eqref{eq:spotcov_estimator_old} and \eqref{eq:spotcov_estimator_old_beta}. Furthermore, the estimator can be viewed as a natural multivariate extension of the univariate candlestick-based volatility estimators, and hence my approach aims to generalize these efficiency gains to the multivariate setting.

\begin{remark}
    At a first glance, it might be surprising that moments of high and low returns are included in the covariance estimator just like the regular returns. In this regard, \cite{rogers2006correlation} show that the correlations of range and asymmetry variables are smooth and nonlinear functions of the correlation coefficient of the underlying price process. Thus this suggests that the cross moments of candlestick returns may contain valuable information about the underlying correlation structure and thus motivates my approach to incorporate them into the estimation process. 
\end{remark}

%However, there exist a connection between the moments of candlestick returns and the underlying correlation structure. To see this, consider a price process $\boldsymbol{X}_t$ generated by a bivariate Brownian motion with unit volatilities and constant correlation coefficient $\rho$. Denote the correlations between the candletick returns by $\rho_r, \rho_a, \rho_w$. Clearly, it holds that $\rho_r=\rho$. As shown by \cite{rogers2006correlation}, the correlation of range and asymmetry variables ($\rho_a, \rho_w$) are smooth and nonlinear functions of $\rho$. Figure \ref{fig:cross_moments} illustrates this relationship. Notably, $\rho_w$ is symmetric around zero, while $\rho_a$ closely follows $45$-degree line. This figure thus suggests that the cross moments of candlestick returns may contain valuable information about the underlying correlation structure and thus motivates our approach to incorporate them into the estimation process.

% \begin{figure}
%     \centering
%     \includegraphics[width=0.88\linewidth]{Figures/corr_func.pdf}
%     \caption{\textbf{Cross moments of candlestick returns:} The figure illustrates the relationship between the correlation of open-close returns $\rho_r$ (blue) and that of asymmetry $\rho_a$ (red) and range $\rho_w$ (green) variables, respectively. The underlying price process is assumed to be a bivariate Brownian motion with unit volatilities and constant correlation coefficient $\rho$.}
%     \label{fig:cross_moments}
% \end{figure}

\subsection{Determining the Weights via Risk Minimization} \label{sec:optimal_weight}
The estimator in Equation \eqref{eq:spotcov_estimator}, and by extension Equation \eqref{eq:spotbeta_estimator}, depends crucially on the choice of weights $\lambda$. Following a decision-theoretic approach, I determine the weighting scheme by minimizing a well-defined risk function. Specifically, the risk of an estimator $\widehat{\boldsymbol{c}}_{n,t}(\lambda)$ is defined as the conditional expectation of its loss:
\begin{equation} \label{eq:risk_function}
    R(\lambda; \boldc_{t}) \equiv \mathbb{E}\left [ L(\lambda; \boldc_{t}) \mid \mathcal{F}_{t} \right ].
\end{equation}
where $L(\lambda; \boldc_{t})$ is a loss function that measures the estimation error of $\widehat{\boldsymbol{c}}_{n,t}(\lambda)$ relative to the true spot covariance $\boldc_t$. I employ the following quadratic loss function:
\begin{equation}
    %L(\boldLambda; {\boldsymbol{\sigma}}_t) = \left \| \boldsigma_t^{-1} (\widehat{\boldsymbol{c}}_{n,t} - \boldc_t) \boldsymbol{\sigma}_t^{-1\top}  \right \|^{2} \quad \text{or, equivalently} \quad 
    L(\lambda; {\boldc}_t) = \left \| \boldsigma_t^{-1} \widehat{\boldsymbol{c}}_{n,t}(\lambda) \boldsigma_t^{-1\top} - \mathbf{I} \right \|^{2}
\end{equation}
where $||\cdot||$ denotes the Frobenius norm and $\mathbf{I}$ is the identity matrix. The term $\boldsymbol{\sigma}_t^{-1}\,\widehat{\boldsymbol{c}}_{n,t}\,\boldsymbol{\sigma}_t^{-1\top} $ indicates the multiplicative estimation error of the spot covariance matrix relative to the identity, naturally inducing a scale-invariant loss function. This property is particularly desirable, as covariance matrices represent scale parameters. In the univariate setting, the loss reduces to
\[
\left( \frac{\widehat{c}_{n,t}}{c_t} - 1 \right)^{2},
\]
which corresponds to the relative squared error that has been previously analyzed in the context of volatility estimation with candlestick data by \cite{li2024reading}.

\begin{remark}
    The determination of the weights $\lambda$ directly targets the full covariance matrix $\boldc_t$ rather than the spot beta $\beta_t$ which is the main focus in a regression context. A reader might therefore ask why covariance-optimal weights is considered here. I deliberately adopt this approach because targeting covariance matrix allows for unified treatment of all spot quantities in the regression representation in Equation \eqref{eq:reg_rep} as they are nonlinear functions of the covariance matrix. In other words, this approach yields a single set of weights that can be used for estimating all spot quantities simultaneously. Furthermore, in Appendix \ref{app:loss_decomp}, I show that this loss function can be decomposed into three components that can be directly associated with the estimation errors of the spot regression parameters. This suggests that minimizing the overall loss effectively balances the trade-offs between the estimation errors of different parameters. Indeed, simulation results in Section \ref{sec:simulation} show that the proposed approach leads to more reliable inference for betas, serving as ex-post supportive evidence. In addition, from a practical perspective, this approach yields a closed-form solution for the weights, which greatly simplifies the computation.
\end{remark}

\subsubsection{Asymptotic Approximation for the Risk Function} \label{sec:asymptotic_risk}
$R(\lambda; \boldc_{t})$ indicates the exact finite sample risk of an estimator. For a given $\lambda$, the risk depends on the joint distribution of $(\boldr_i, \bolda_i, \boldw_i)$ which in turn is determined by the law of $(\boldsymbol{b}, \boldsigma, \boldW)$. Since this law is unknown, the estimation of $\lambda$ with direct minimization of $R(\lambda; \boldsigma_{t})$ is infeasible in practice. To circumvent this issue, I first establish an asymptotic approximation of the multiplicative estimation error \(\boldsymbol{\sigma}_t^{-1}\widehat{\boldsymbol{c}}_{n,t}\boldsymbol{\sigma}_t^{-1\top}\) using coupling arguments (see, e.g., \cite{jacod2021volatility}). This characterization is subsequently extended to the loss function and, in turn, to the risk. The resulting asymptotic risk functional is then used to determine the optimal weights \(\lambda\).

The assumption below gathers a set of regularity conditions, required for the asymptotic approximation results. 

\begin{assumption} \label{asmp:spotcov_asmp}
    Suppose that $\boldsymbol{X}_t$ has the form in Equation \eqref{eq:itosemi} and there exists a sequence $(T_m)_{m \geq 1}$ of stopping times increasing to infinity and the following conditions hold for each $m \geq 1$:
    \begin{itemize}
        \item[(i)] $\|\boldsymbol{b}_t\| + \|\boldsymbol{\sigma}_t\| + \|\boldsigma_t^{-1}\| \leq K_m$ for some constant $K_m$ for all $t \in [0, T_m]$;
        \item[(ii)] $\mathbb{E}\left[ \|\boldsigma_{t \wedge T_m}-\boldsigma_{s\wedge T_m} \|^2 \right] \leq K_m |t-s|$ for all $t, s \in [0, T_m]$.
    \end{itemize}
\end{assumption}   

These conditions are quite standard in the high-frequency econometrics literature; see, for example, \cite{jacod2012discretization} for further details. The first part of the Assumption \ref{asmp:spotcov_asmp} implies local boundedness of the drift, volatility, and inverse volatility processes while the second part ensures a degree of smoothness in the volatility process, specifically called as locally $1/2$-Hölder continuous. These assumptions are sufficiently general to accommodate a wide range of volatility dynamics, including volatility jumps, leverage effects, and intraday seasonality. For example, they are satisfied if the volatility process is itself an Itô semimartingale or long-memory process driven by fractional Brownian motion.

The following proposition describes the key approximation result for the estimation error \(\boldsymbol{\sigma}_t^{-1}\widehat{\boldsymbol{c}}_{n,t}\boldsymbol{\sigma}_t^{-1\top}\). The proof is deferred to Appendix \ref{app:prop_1_proof}.

\begin{proposition} \label{prop:spotcov_coupling}
    Suppose that Assumption \ref{asmp:spotcov_asmp} holds. Fix any $t \in [0,T]$. For any $k \geq 1$ and $\lambda$, the following holds as $\Delta_n \to 0$:

    \begin{equation} \label{eq:spotcov_coupling}
        \left\| \boldsymbol{\sigma}_t^{-1}\widehat{\boldsymbol{c}}_{n, t}(\lambda)\boldsymbol{\sigma}_t^{-\top} -  U_{n,t}(\lambda) \right\| = o_p(1)
    \end{equation}
    where
    \begin{equation}\label{eq:U_nt}
        \begin{aligned} 
            U_{n,t}(\lambda) &= \frac{1}{k}\sum_{i\in \mathcal{I}_{n,t}} \Big\{\lambda_1 \boldzeta_{i,r} \boldzeta_{i,r}^{\top} + \lambda_2 \boldzeta_{i,a} \boldzeta_{i,a}^{\top} + \lambda_3 \boldzeta_{i,w} \boldzeta_{i,w}^{\top} \\
            &+ \lambda_4 (\boldzeta_{i,r} \boldzeta_{i,a}^{\top} + \boldzeta_{i,a} \boldzeta_{i,r}^{\top}) + \lambda_5(\boldzeta_{i,r} \boldzeta_{i,w}^{\top} + \boldzeta_{i,w} \boldzeta_{i,r}^{\top}) + \lambda_6(\boldzeta_{i,a} \boldzeta_{i,w}^{\top} + \boldzeta_{i,w} \boldzeta_{i,a}^{\top})\Big\}
        \end{aligned}
    \end{equation}
    and, for any $i\in \mathcal{I}_{n,t}$,
    \begin{equation}\label{eq:coupling_vars}
    \begin{array}{rcl}
        \boldsymbol{\zeta}_{i, r} &\equiv&  
         \frac{\boldW_{i\Delta_n} - \boldW_{(i-1)\Delta_n}}{\sqrt{\Delta_n}} \\
         \boldsymbol{\zeta}_{i, a} &\equiv& 
          \boldvarrho_{t}^{-1}\underset{\tau \in \mathcal{T}_i}{\sup \:} \boldvarrho_{t} \left(\frac{\boldW_{\tau} - \boldW_{(i-1)\Delta_n}}{\sqrt{\Delta_n}}\right) + \boldvarrho_{t}^{-1}\underset{\tau \in \mathcal{T}_i}{\inf \:} \boldvarrho_{t}   \left(\frac{\boldW_{\tau} - \boldW_{(i-1)\Delta_n}}{\sqrt{\Delta_n}}\right) - \left(\frac{\boldW_{i\Delta_n} - \boldW_{(i-1)\Delta_n}}{\sqrt{\Delta_n}}\right) \\
         \boldsymbol{\zeta}_{i, w} &\equiv& 
          \boldvarrho_{t}^{-1}\underset{\tau \in \mathcal{T}_i}{\sup \:} \boldvarrho_{t}   \left(\frac{\boldW_{\tau} - \boldW_{(i-1)\Delta_n}}{\sqrt{\Delta_n}}\right) - \boldvarrho_{t}^{-1}\underset{\tau \in \mathcal{T}_i}{\inf \:} \boldvarrho_{t}   \left(\frac{\boldW_{\tau} - \boldW_{(i-1)\Delta_n}}{\sqrt{\Delta_n}}\right) \\
    \end{array}
    \end{equation}
    with $\boldvarrho_t$ being the square root of spot correlation matrix $\boldsymbol{\rho}_t$, i.e., $\boldsymbol{\rho}_t = \boldvarrho_t \boldvarrho_t^\top$. All $\inf$ and $\sup$ operators are applied element-wise.
\end{proposition}

Proposition \ref{prop:spotcov_coupling} establishes that the multiplicative estimation error of the spot covariance estimator $\widehat{\boldsymbol{c}}_{n, t}(\lambda)$ is asymptotically approximated by the random matrix $U_{n, t}$ in probability as $\Delta_n \to 0$. The structure of $U_{n, t}$ mirrors that of the original estimator, with the candlestick returns replaced by the variables $\boldzeta_i$s. These variables are functions of Brownian motion $\boldW$ and square root of spot correlation matrix $\boldvarrho_t$. In addition, the approximation error ($o_p(1)$ term) captures nonparametric biases arising from stochastic volatility and drift. Put differently, in the limiting case of constant volatility and vanishing drift, the relationship in Proposition \ref{prop:spotcov_coupling} holds exactly rather than approximately.

Using Proposition~\ref{prop:spotcov_coupling} and the conditions stated therein, analogous approximations apply to the loss and risk functions. Since the loss function is continuous, the continuous mapping theorem implies
\begin{equation} \label{eq:asymptotic_loss}
    L(\lambda; \boldc_t) = \left\| U_{n, t}(\lambda) - \mathbf{I}\right\|^2  +  o_p(1).
\end{equation}
This relation can similarly be extended to the risk function by taking the conditional expectation of both sides given $\mathcal{F}_t$:
\begin{equation} \label{eq:asymptotic_risk}
    {R}(\lambda; \boldc_t) = \underbrace{\mathbb{E}\left[ \left\| U_{n, t}(\lambda) - \boldsymbol{I} \right\|^2 \mid \mathcal{F}_{t} \right]}_{\widetilde{R}(\lambda; \boldrho_{t})} + o_p(1)
\end{equation}
where $\widetilde{R}(\lambda; \boldrho_t)$ is the asymptotic risk of the estimator $\widehat{\boldsymbol{c}}_{n, t}$.\footnote{This step requires an additional uniform integrability condition to interchange the limit and expectation operators. This condition is satisfied under the assumptions stated in Assumption \ref{asmp:spotcov_asmp}.} Importantly, I switch to notation $\widetilde{R}(\lambda; \boldrho_t)$ to emphasize that the asymptotic risk depends on the spot correlation matrix $\boldrho_t$ rather than the spot covariance matrix $\boldc_t$. This is because the risk formula involves a term, $U_{n,t}$, which depends on $\boldvarrho_t$, i.e., the square root of the spot correlation matrix. The use of the Frobenius norm in the final calculation effectively considers the product $\boldvarrho_t \boldvarrho_t^\top$, making the risk solely dependent on $\boldrho_t$.

As a result, direct minimization of the asymptotic risk $\widetilde{R}(\lambda; \boldrho_t)$ with respect to $\lambda$ to obtain the optimal weights is infeasible in practice, as $\boldrho_t$ is unknown to the researcher. In the univariate case, by contrast, the term $U_{n,t}$ does not depend on any unknown parameters and thus the distribution of $U_{n,t}$ is pivotal. This allows for direct computation of the asymptotic risk and thus the derivation of optimal weights via direct risk minimization, see Theorem 1 in \citet{li2024reading, bollerslev2024optimal}. This is because the scaling by $\boldsymbol{\varrho}_t$ and its inverse cancels, as they become scalars, in Equation \eqref{eq:coupling_vars}. Hence, once can characterize the distribution of $U_{n,t}$ through simulations of Brownian motion functionals. This simplification does not hold in the multivariate setting, and thus the estimation problem studied in this paper requires more intricate treatment than the univariate case.

%In particular, the non-pivotal nature of the distribution of $U_{n,t}$ complicates the direct computation of the asymptotic risk and thus the derivation of optimal weights. To see this, consider the structure of $U_{n,t}$, which is derived from the variables $(\boldsymbol{\zeta}_{i,r}, \boldsymbol{\zeta}_{i,a}, \boldsymbol{\zeta}_{i,w})$. By the scaling property of Brownian motion, $\boldsymbol{\zeta}_{i,r}$ is standard normally distributed and therefore naturally pivotal. However, the remaining variables, $\boldsymbol{\zeta}_{i,a}$ and $\boldsymbol{\zeta}_{i,w}$, depend on $\boldsymbol{\varrho}_t$ in a nontrivial manner. This dependence arises because the supremum and infimum operators are nonlinear, preventing the scaling by $\boldsymbol{\varrho}_t$ and its inverse from canceling. Consequently, the distribution of $U_{n,t}$ is generally not free from $\boldsymbol{\varrho}_t$. 

%In the univariate case, by contrast, the terms $\boldsymbol{\zeta}_{i,a}$ and $\boldsymbol{\zeta}_{i,w}$ reduce to functions of standard Brownian motion increments. This is because the scaling by $\boldsymbol{\varrho}_t$ and its inverse cancels, as they become scalars, allowing the distribution of these terms to be fully characterized through simulations of Brownian motion functionals. Consequently, the asymptotic risk can be computed directly from this simulated distribution. As such, the derivation of optimal weights is greatly simplified via direct risk minimization as in \citet{li2024reading, bollerslev2024optimal}. 

\subsubsection{Average Risk} \label{sec:average_risk}
To address the complication discussed above, I propose an alternative approach for selecting the weights $\lambda$ based on the average risk $\bar{R}(\lambda)$, defined as the integrated asymptotic risk $\widetilde{R}(\lambda; \boldsymbol{\rho}_t)$ over the parameter space of $\boldsymbol{\rho}_t$. This approach effectively marginalizes over the unknown correlation structure, thereby enabling the derivation of ``optimal'' weights without requiring explicit knowledge of $\boldsymbol{\rho}_t$.\footnote{Here optimality is used in the context of minimizing the average risk which is specifically defined in Equation \eqref{eq:average_risk}.} This perspective is directly inspired by the statistical decision theory literature, see for example \cite{lehmann2006theory}, which advocates the use of integrated risk functions to accommodate nuisance parameters. 

Let $\mathcal{P}$ denote the parameter space of positive semidefinite correlation matrices. Define the average risk $\bar{R}(\lambda)$ as the integrated asymptotic risk
\begin{equation} \label{eq:average_risk}
    \bar{R}(\lambda) \equiv \int_{\mathcal{P}} \widetilde{R}(\lambda; \boldsymbol{\rho}_t) d\boldsymbol{\rho}_t.
\end{equation}
Put differently, $\bar{R}(\lambda)$ is averaging over the interval $(-1,1)$ as the setup has two assets and single correlation parameter.

For any fixed $\lambda$, the mapping $\boldsymbol{\rho} \mapsto \widetilde{R}(\lambda; \boldsymbol{\rho})$ can be computed via Monte Carlo simulation. Specifically, conditional on $\boldsymbol{\rho}_t$, the distribution of $U_{n,t}$ is fully characterized by functionals of Brownian motion defined. Thus, the risk $\widetilde{R}(\lambda; \boldsymbol{\rho}_t)$ can be computed by simulating these functionals and evaluating the corresponding loss function. Consequently, the average risk $\bar{R}(\lambda)$ depends only on $\lambda$, and the optimal weights can be obtained by solving
\begin{equation} 
    \lambda^* = \arg\min_{\lambda} \bar{R}(\lambda).
\end{equation}

The solution to this optimization problem is available in closed form and resembles a least-squares formula:
\begin{equation} \label{eq:optimal_lambda}
    \lambda^* = \left(\int_{\mathcal{P}} {\mathbb{E}}[\boldsymbol{\Pi}^{\top}\boldsymbol{\Pi} | \mathcal{F}_t]d\boldrho_t \right)^{-1} \left(\int_{\mathcal{P}} {\mathbb{E}}[\boldsymbol{\Pi}^{\top}\boldsymbol{y} | \mathcal{F}_t] d\boldrho_t \right)
\end{equation}
where $\boldsymbol{\Pi}$ is stacking the vectorized version of outer products of limiting variables and $\boldsymbol{y}$ is the vectorized version of the identity matrix $\mathbf{I}$. More explicitly, $\boldsymbol{y} = \vech(\mathbf{I})$ is $3 \times 1$ vector and
\begin{equation}
\begin{aligned}
    \boldsymbol{\Pi} \equiv \bigg[
    & \tfrac{1}{k} \sum_i \vech(\boldzeta_{i,r} \boldzeta_{i,r}^{\top}) \;\;
      \tfrac{1}{k} \sum_i \vech(\boldzeta_{i,a} \boldzeta_{i,a}^{\top}) \;\;
      \tfrac{1}{k} \sum_i \vech(\boldzeta_{i,w} \boldzeta_{i,w}^{\top}) \cdots \\
    & \tfrac{1}{k} \sum_i \vech(\boldzeta_{i,r} \boldzeta_{i,a}^{\top} + \boldzeta_{i,a} \boldzeta_{i,r}^{\top}) \;\;
      \tfrac{1}{k} \sum_i \vech(\boldzeta_{i,r} \boldzeta_{i,w}^{\top} + \boldzeta_{i,w} \boldzeta_{i,r}^{\top}) \;\;
      \tfrac{1}{k} \sum_i \vech(\boldzeta_{i,a} \boldzeta_{i,w}^{\top} + \boldzeta_{i,w} \boldzeta_{i,a}^{\top}) \bigg]
\end{aligned}
\end{equation}
is a $3 \times 6$ matrix and each column corresponds to the half-vectorized version of the outer products of the limiting variables defined in Equation \eqref{eq:coupling_vars}.

The resulting spot covariance estimator can be denoted as $\widehat{\boldc}_{n,t}(\lambda^*)$ and the corresponding estimators for spot beta, idiosyncratic variance, and systematic variance can be obtained by substituting $\lambda^*$ into Equation \eqref{eq:spotbeta_estimator}.
\begin{equation}
    \widehat{\beta}_{n,t}(\lambda^*) = \frac{\widehat{c}_{12,t}(\lambda^*)}{\widehat{c}_{11,t}(\lambda^*)}, \quad
    \widehat{\varsigma}_{n,t}(\lambda^*) = \widehat{c}_{22,t}(\lambda^*) - \frac{\widehat{c}_{12,t}^2(\lambda^*)}{\widehat{c}_{11,t}(\lambda^*)}, \quad \text{and} \quad
    \widehat{\nu}_{n,t}(\lambda^*) = \widehat{c}_{11,t}(\lambda^*).
\end{equation}

In practice, the integrals in $\lambda^*$ formula (Equation \eqref{eq:optimal_lambda}) can be computed by simulation. Specifically, one can draw many realizations of $\boldsymbol{\rho}$ uniformly from $\mathcal{P}$, calculate the corresponding conditional expectations via Monte Carlo simulation, and replace the integrals with sample averages. The estimator is then obtained by substituting these into Equation \eqref{eq:optimal_lambda}. Importantly, this numerical procedure takes estimation window size $k$ as a given constant, consistent with this paper's asymptotic framework.

Equation \eqref{eq:optimal_lambda_results} presents the optimal weights $\lambda^*$ for different values of $k$. There are several noteworthy observations regarding these results. First, the optimal weights for the first two components, $\lambda_1$ and $\lambda_2$, are notably larger than those for the remaining components, suggesting that the open-close returns and asymmetry variables play a more dominant role in the estimates. Surprisingly, the optimal weight for the range variable, $\lambda_3$, is near zero across all settings. The reason is that the range variable loses the sign information of the underlying correlation structure and is heavily biased for negative correlations. Consequently, the estimation process substantially downweights it. On the other hand, the cross terms mechanically receive exact zero weight since $(\boldr_i, \bolda_i, \boldw_i)$ are asymptotically orthogonal to each other.

\begin{equation} \label{eq:optimal_lambda_results}
    \lambda^* = \begin{cases}
        (0.438, ~~ 1.523, ~~ \approx 0, ~~ 0, ~~ 0, ~~ 0)^{\top} & \text{if  }k=5, \\[0.5em]
        (0.488, ~~ 1.648, ~~ \approx 0, ~~ 0, ~~ 0, ~~ 0)^{\top} & \text{if  }k=10, \\[0.5em]
        (0.526, ~~ 1.692, ~~ \approx 0, ~~ 0, ~~ 0, ~~ 0)^{\top}  & \text{if  }k=20.
    \end{cases}
\end{equation}

\begin{remark}
    When defining the average risk, the integration over $\mathcal{P}$ is performed uniformly which also induces a uniform marginal distribution for the correlation parameter. This may appear simplistic and a natural concern can arise regarding this specification. Indeed alternative choices including more flexible distributional assumptions or those informed by the historical data and expert judgment could be certainly considered. Nevertheless, the asymptotic risk comparisons presented in Section \ref{sec:risk_comparison} reveals that this approach performs nearly as well as oracle estimators, leaving limited potential improvement through alternative specifications. Moreover, solving this optimization problem for $\lambda$ does not inherently depend on uniform integration. As such, the optimal weights can be obtained for any given marginal distribution of $\boldrho$ following the same computational steps.
\end{remark}

\begin{remark}
    An estimator of similar form, a linear combination of $\boldr_i \boldr_i^{\top}$ and $\bolda_i \bolda_i^{\top}$, was previously studied by \cite{rogers2008estimating} for the purpose of estimating the correlation between two Brownian motions. Their analysis assumes that the price process follows a scaled Brownian motion with constant volatility and determines the weights by minimizing the variance of the estimation error, subject to an unbiasedness constraint under the assumption of zero correlation. By contrast, my approach is formulated within a general Itô-semimartingale framework which accommodates stochastic volatility and selects the weights by minimizing the risk without imposing restrictions on the correlation level. In addition, the next section develops an inference procedure that enables hypothesis testing for spot betas, whereas \citet{rogers2008estimating} focus exclusively on point estimation. Beyond these methodological differences, the Monte Carlo experiments in Section \ref{sec:simulation} and the empirical application in Section \ref{sec:empirics} provide a comprehensive evaluation of candlestick-based estimators in a multivariate setting, which is also absent from the prior literature.
\end{remark}

\section{Inference on Spot Betas} \label{sec:inference}
Spot betas are fundamental objects of interest in asset pricing and risk management, yet inference on these quantities remains an underexplored area in the candlestick-based estimation literature. Existing work has focused on inference for spot volatility functionals (see, e.g., \cite{li2024reading}; \cite{bollerslev2024optimal}). This section fills that gap. Specifically, I develop a hypothesis test for the spot beta, derive its limiting null distribution, and provide a feasible procedure for computing critical values.

\subsection{Test Statistic}
For a fixed time point $t\in [0,T]$, and hypothesized value $\beta_0 \in \mathbb{R}$, consider testing the null hypothesis 
$$
    H_0: \beta_t = \beta_0 \quad \text{against} \quad  H_1: \beta_t \neq \beta_0
$$

Building on the candlestick-based estimator,
I propose the following studentized test statistic:
$$
    \widehat{T}_n(\lambda) = \frac{\sqrt{k-1}\left (\widehat{\beta}_{n,t}(\lambda) - \beta_0\right )}{\sqrt{\widehat{\varsigma}_{n,t}(\lambda) / \widehat{\nu}_{n,t}(\lambda)}}
$$
where $\widehat{\beta}_{n,t}(\lambda)$, $\widehat{\nu}_{n,t}(\lambda)$, and $\widehat{\varsigma}_{n,t}(\lambda)$ are candlestick-based estimates for some $\lambda$. Therefore, it encompasses both the return-based statistic of \cite{bollerslev2024optimalspotregressions}, recovered at $\lambda = (1, 0, ..., 0)^\top$, and the candlestick-based statistic associated with the risk-minimizing weights $\lambda^*$ derived in Section \ref{sec:estimation}.

To study the asymptotic properties of the test statistic, it is useful to establish an asymptotic approximation to $\widehat{T}_n(\lambda)$, which can be achieved by exploiting Proposition \ref{prop:spotcov_coupling} and the continuous mapping theorem. This result is summarized in the following proposition. Proof is provided in Appendix \ref{app:prop_2_proof}. Similar to the estimation problem, fixed-$k$ asymptotics is adopted here.

\begin{proposition} \label{prop:test_coupling}
    Suppose that the conditions of Proposition 1 hold and $\boldsymbol{X}_t$ follows Equation \eqref{eq:reg_rep}. Then, for any fixed $k \geq 2$ and $\lambda$, the following holds as $\Delta_n \to 0$:
    $$
        |\widehat{T}_n(\lambda) - \widetilde{T}_n(\lambda)| = o_p(1) \quad \text{where} \quad \widetilde{T}_n(\lambda) \equiv \frac{-\sqrt{k-1}[U_{n, t}^{-1}]_{12}}{\sqrt{[U_{n, t}^{-1}]_{11}[U_{n, t}^{-1}]_{22}-[U_{n, t}^{-1}]_{12}^2}}.
    $$
\end{proposition}

This proposition shows that $\widehat{T}_n(\lambda)$ can be asymptotically approximated by $\widetilde{T}_n(\lambda)$, which is a function of the matrix $U_{n,t}$ defined in Proposition \ref{prop:spotcov_coupling}. As such, the distributional properties of the limiting variable $\widetilde{T}_n(\lambda)$ can be exploited to compute the critical values for the test. Next, I turn to this problem.

%For the return-based statistic, i.e., when $\lambda = (1, 0, 0)^{\top}$, the distribution of $\widetilde{T}_n(\lambda)$ admits a closed-form Student's $t$ distribution with $k-1$ degrees of freedom as shown by \cite{bollerslev2024optimalspotregressions}. For general $\lambda$, and in particular for the candlestick-based weights $\lambda^*$, no such closed-form distribution is available. I turn to this problem next. 

\subsection{Critical Values}
For a given nominal level $\alpha \in (0,1)$, critical values for the test can be described formally as follows:
\begin{equation}
    \mathbb{P}\!\left( B^{-}_{\alpha} < \widetilde{T}_n(\lambda) < B^{+}_{\alpha} \,\middle|\, H_0 \right) = 1 - \alpha.
\end{equation}
where constants $B^{+}_{\alpha}$ and $B^{-}_{\alpha}$ are the upper and lower critical values for the spot beta test, respectively. The test then rejects the null hypothesis $H_0$ if $\widehat{T}_n(\lambda)$ falls outside the interval $(B^{-}_{\alpha}, B^{+}_{\alpha})$. Since infinitely many pairs can satisfy this condition, one can select the highest density interval (HDI), which is the narrowest interval containing the desired probability mass. By construction, this procedure delivers a test with asymptotically correct size. 

The computation of the critical values $(B^{-}_{\alpha}, B^{+}_{\alpha})$ depends on whether the distribution of $\widetilde{T}_n(\lambda)$ under $H_0$ is pivotal or not. For instance, for the return-based statistic, i.e., when $\lambda = (1, 0, ...,0)^{\top}$, the distribution of $\widetilde{T}_n(\lambda)$ admits a closed-form Student's $t$ distribution with $k-1$ degrees of freedom as shown by \cite{bollerslev2024optimalspotregressions}. For general $\lambda$, and in particular for the candlestick-based weights $\lambda^*$, $U_{n,t}(\lambda)$ depends on the unknown correlation parameter $\boldrho_t$ in a non-trivial manner, and thus the distribution of $\widetilde{T}_n(\lambda)$ under $H_0$ is generally non-pivotal. 

As a special case, the pivotal property is restored when $\beta_0 = 0$ in the null hypothesis, consequently its distribution can be simulated directly. This is because under $H_0: \beta_t = 0$, the covariance matrix $\boldc_t$ defined in Equation \eqref{eq:regression_covariance} simplifies to a diagonal matrix and consequently the $\boldvarrho_{t}$ terms in the coupling returns provided in Equation \eqref{eq:coupling_vars} cancel out, leaving the limiting variable as a function of the Brownian motion functionals only. Therefore, the quantiles of $\widetilde{T}_n(\lambda)$ can be simulated which in turn allows for feasible inference on spot betas using the proposed estimator. Indeed, this case yields an asymptotically exact size test.

I now consider the general case, i.e. $H_0: \beta_t = \beta_0$ with $\beta_0 \neq 0$. When the hypothesized value is non-zero, the distribution of $\widetilde{T}_n(\lambda)$ is generally non-pivotal, and hence direct simulation is not feasible. To restore the feasibility of inference, I follow a sup-t test approach by constructing critical values that uniformly control size across the parameter space of $\boldrho_t$. 

Specifically, let the pair $(B^{-}_{\alpha}(\boldrho), B^{+}_{\alpha}(\boldrho))$ be the HDI of $\widetilde{T}_n(\lambda)$ given some $\boldrho \in \mathcal{P}$ and significance level $\alpha$. Define
\begin{equation}
    B^{-}_{\alpha} = \inf_{\boldrho \in \mathcal{P}} B^{-}_{\alpha}(\boldrho) \quad \text{and} \quad B^{+}_{\alpha} = \sup_{\boldrho \in \mathcal{P}} B^{+}_{\alpha}(\boldrho).
\end{equation}
By construction, $(B^{-}_{\alpha}, B^{+}_{\alpha})$ contains every conditional HDI and, consequently, this guarantees that
\begin{equation}
    \inf_{\boldrho \in \mathcal{P}} \mathbb{P}\!\left( B^{-}_{\alpha} < \widetilde{T}_n(\lambda) < B^{+}_{\alpha} \,\middle|\, H_0 \right) \geq 1 - \alpha.
\end{equation}

The cutoffs $B^{-}_{\alpha}$ and $B^{+}_{\alpha}$ are the envelopes of the conditional HDIs over $\mathcal{P}$ and therefore deliver uniform asymptotic size control: the test rejects $H_0$ with probability at most $\alpha$. This uniformity comes at the cost of conservativeness, but it requires no preliminary estimation step and critical values can be tabulated in a single simulation step. Simulation results in Section \ref{sec:simulation} indicate that the resulting test is not overly conservative and delivers good power properties.

Table \ref{tab:beta_confint_combined} reports the critical values for $H_0: \beta_t = 0$ and $H_0: \beta_t = \beta_0$ with $\beta_0 \neq 0$, respectively, for range of $k$ and $\alpha$ combinations, computed via $10{,}000$ Monte Carlo simulations. For comparison, I also report the analogous critical values for the return-based beta estimator, which are obtained in closed form from the Student's $t$-distribution following \cite{bollerslev2024optimalspotregressions}. To facilitate a direct assessment of efficiency, I further report the interval widths $B^{+}_{\alpha} - B^{-}_{\alpha}$, for both estimators. 

\begin{table}[t!]
    \centering
    \renewcommand{\arraystretch}{1.3}
    \setlength{\tabcolsep}{9pt}
    \caption{\textbf{Critical values for }$H_0: \beta_t = \beta_0$ : This table reports the critical values for the hypothesis test of $H_0: \beta_t = \beta_0$ at significance levels $\alpha \in \{5\%, 10\%\}$ and for different local window sizes $k \in \{5, 10, 20\}$. The left panel shows the critical values for the return-based beta estimator derived from the Student's $t$-distribution with $k-1$ degrees of freedom. The right panel presents the critical values for the candlestick-based beta estimator computed via $10{,}000$ Monte Carlo simulations. The interval width is defined as $B^{+}_{\alpha} - B^{-}_{\alpha}$.}
    \smallskip
    \label{tab:beta_confint_combined}
    \adjustbox{max width=\textwidth}{
    \begin{tabular}{c ccc c ccc c ccc}
    \toprule
    & \multicolumn{3}{c}{\textbf{Return}} 
    & & \multicolumn{3}{c}{\textbf{Candlestick ($\beta_0=0$)}} 
    & & \multicolumn{3}{c}{\textbf{Candlestick ($\beta_0\neq0$)}} \\
    \cmidrule{2-4} \cmidrule{6-8} \cmidrule{10-12}
    $\boldsymbol{k}$ 
    & $B^-_\alpha$ & $B^+_\alpha$ & Width 
    & & $B^-_\alpha$ & $B^+_\alpha$ & Width 
    & & $B^-_\alpha$ & $B^+_\alpha$ & Width \\
    \midrule

    \multicolumn{11}{l}{\textbf{Panel A:} $\alpha=5\%$} \\
    \midrule
    5  & -2.776 & 2.776 & 5.552 & & -1.657 & 1.415 & 3.072 & & -1.695 & 1.636 & 3.332 \\
    10 & -2.262 & 2.262 & 4.524 & & -1.430 & 1.460 & 2.890 & & -1.609 & 1.568 & 3.177 \\
    20 & -2.093 & 2.093 & 4.186 & & -1.429 & 1.383 & 2.812 & & -1.567 & 1.563 & 3.130 \\

    \midrule
    \multicolumn{11}{l}{\textbf{Panel B:} $\alpha=10\%$} \\
    \midrule
    5  & -2.132 & 2.132 & 4.264 & & -1.204 & 1.262 & 2.466 & & -1.375 & 1.336 & 2.711 \\
    10 & -1.383 & 1.383 & 2.766 & & -1.203 & 1.199 & 2.402 & & -1.338 & 1.337 & 2.674 \\
    20 & -1.328 & 1.328 & 2.656 & & -1.103 & 1.244 & 2.347 & & -1.372 & 1.364 & 2.737 \\

    \bottomrule
    \end{tabular}}
\end{table}

These results indicate that the critical values are tighter for the candlestick-based beta estimator compared to the return-based estimator, across all values of $k$ and $\alpha$. Specifically, the difference in interval widths becomes more pronounced as $k$ decreases. This highlights the usefulness of my estimator, particularly in limited data scenarios.

\section{Monte Carlo Evidence} \label{sec:simulation}
In this section, I examine the performance of the candlestick-based estimator and the test through Monte Carlo experiments. The analysis proceeds in two parts. First, I investigate the efficiency of the proposed estimator by comparing its asymptotic risk against an oracle estimator and a return-based estimator. Second, I evaluate the power of the hypothesis test for spot betas introduced in Section \ref{sec:inference}.

\subsection{Efficiency of the Spot Covariance Estimator} \label{sec:risk_comparison}
To assess the efficiency of my approach, I compare the asymptotic risk of my estimator against two natural benchmarks.\footnote{It is important to note that the asymptotic risk $\widetilde{R}(\lambda; \boldrho_t)$ can be computed via simulations for any $\lambda$ and $k$ if the correlation structure is known.} The first benchmark is an infeasible oracle estimator that minimizes asymptotic risk under knowledge of the true correlation structure. This estimator serves as a theoretical upper bound for the performance of any feasible estimator, as it uses population information. The second benchmark is the return-based estimator, which relies solely on (open-to-close) returns. All three estimators belong to the class defined in Equation \eqref{eq:spotcov_estimator}, differing only in the choice of $\lambda$: the return-based estimator sets $\lambda=(1,0, ..., 0)^{\top}$, the oracle sets $\lambda=\operatorname{argmin}_{\lambda}\,\widetilde{R}(\lambda,\boldsymbol{\rho}_t)$, and the proposed estimator sets $\lambda=\lambda^*$, minimizing the average risk across correlation structures.

Table \ref{tab:asymptotic_risk} reports the asymptotic risk values for various configurations. I consider 3 different levels of correlation, $\rho \in \{0, 0.2, 0.6\}$.\footnote{These numbers represent the $10th, 50th, 90th$ percentiles of cross-section of pairwise correlations among the S\&P 500 stocks.} Moreover, the size of the local estimation window is set to $k \in \{5, 10, 20\}$. All results are based on 10,000 Monte Carlo simulations.

\begin{table}[t!]
    \centering
    \caption{\textbf{Asymptotic Risk of Estimators for Different $\rho$ Values}: The table presents the asymptotic risk of three estimators: the return-based estimator, the candlestick-based estimator, and an oracle estimator that minimizes asymptotic risk with knowledge of the true correlation structure. The results are shown for various local window sizes \( k \in \{5, 10, 20\} \) and correlation levels \( \rho \in \{0, 0.2, 0.6\} \). The asymptotic risk is computed via Monte Carlo simulations.}
    \label{tab:asymptotic_risk}
    \renewcommand{\arraystretch}{1.2} % slightly taller for readability
    \setlength{\tabcolsep}{12pt} % adjust horizontal spacing
    % adjust vertical spacing
    \smallskip
    \resizebox{0.999\textwidth}{!}{%
    \begin{tabular}{c c cc cc cc}
         \toprule \vspace{-0.40cm} \\
         & & \multicolumn{2}{c}{$\boldsymbol{\rho = 0}$} & \multicolumn{2}{c}{$\boldsymbol{\rho = 0.2}$} & \multicolumn{2}{c}{$\boldsymbol{\rho = 0.6}$} \\
        \cmidrule(lr){3-4} \cmidrule(lr){5-6} \cmidrule(lr){7-8}
        $\boldsymbol{k}$ & Return & Candlestick & Oracle & Candlestick & Oracle & Candlestick & Oracle \\
        \midrule \vspace{-0.25cm}\\
        %1  & 6.119 & 1.136 & 0.941 & 1.140 & 0.961 & 1.161 & 1.016 \\
        5  & 1.212 & 0.427 & 0.384 & 0.429 & 0.384 & 0.440 & 0.391 \\
        10 & 0.599 & 0.240 & 0.222 & 0.241 & 0.221 & 0.249 & 0.224 \\
        20 & 0.297 & 0.128 & 0.119 & 0.128 & 0.118 & 0.134 & 0.119 \vspace{0.25cm}\\
        \bottomrule
    \end{tabular}} \vspace{0.25cm}
\end{table}

These results provide several key insights. First, my estimator achieves asymptotic risk levels that are very close to those of the oracle estimator across all values of $k$ and $\rho$, implying roughly $8-9\%$ efficiency loss. This indicates that the proposed approach effectively addresses complications arising from the distribution of the estimator being non-pivotal, resulting in an estimator that performs nearly as well as the infeasible oracle estimator. Second, the proposed estimator consistently demonstrates a significant reduction in asymptotic risk compared to the return-based estimator, particularly for smaller values of $k$. That is, for $k=5$, the candlestick-based estimator's risk is roughly 65\% lower than that of the return-based estimator, and this difference is about $55\%$ for $k=20$. This suggests that my approach effectively leverages the candlestick prices and yields more efficient estimators. With only $k=5$ observations, the proposed estimator achieves a lower risk than the return-based estimator using $k=10$ observations, demonstrating its superior efficiency in extracting information from narrow time windows. This advantage is particularly valuable in high-frequency event studies, where identification often relies on short time intervals and data are inherently limited.

Overall, this analysis reveals that incorporating candlestick information substantially improves efficiency, providing a powerful and practical alternative to traditional return-based methods. Importantly, my approach performs nearly as well as oracle estimators, leaving limited potential improvement through alternative specifications.

\subsection{Size and Power of the Spot Beta Test}
Next, I evaluate the properties of the hypothesis test for spot betas introduced in Section \ref{sec:inference} with finite sample simulations. 

To do this, I simulate high-frequency price data from a data generating process (DGP), previously studied in \cite{ bollerslev2024optimalspotregressions}.\footnote{This DGP is built on the univariate setup originally proposed by \cite{bollerslev2011estimation}. Later, it is implemented by a number of subsequent studies (e.g., \cite{bollerslev2021fixed,li2024reading}).} Specifically, this DGP assumes a two-factor structure for the market variance process: $\nu_t=V_{1,t}+V_{2,t}$ where $V_{1,t}$ and $V_{2,t}$ follow the processes:
\begin{equation}
\begin{aligned}
& d V_{1, t}=0.0128\left(0.4068-V_{1, t}\right) d t+0.0954 \sqrt{V_{1, t}}\left(\gamma d W_{1, t}+\sqrt{1-\gamma^2} d B_{1, t}\right), \\
& d V_{2, t}=0.6930\left(0.4068-V_{2, t}\right) d t+0.7023 \sqrt{V_{2, t}}\left(\gamma d W_{1, t}+\sqrt{1-\gamma^2} d B_{2, t}\right).
\end{aligned}
\end{equation}
Here, $B_{1,t}$ and $B_{2,t}$ are independent Brownian motions that are also independent of $W_{1,t}$ and $W_{2,t}$. The parameter $\gamma$ is set to $-0.7$, capturing the well-documented leverage effect in financial markets. The coefficients are set so that the first volatility factor is highly persistent, showing $2.5$-month half-life, while the second factor reverts to the mean quickly, with a half-life of a day. This allows the DGP to capture both short-term fluctuations and long-term trends in volatility. 

The idiosyncratic variance $\varsigma_t$ and the spot beta are assumed to follow:
\begin{equation}
    \beta_t = 1 + 0.25 \sin(t)^2, \quad \varsigma_t = \left(1.5 + 0.25 \sin(t)^2\right)\nu_t
\end{equation} 
where $\beta_t$ fluctuates between $1$ and $1.25$ over time and $\varsigma_t$ is set to be proportional to the market variance $\nu_t$. Finally, the price processes are generated using the representation provided in Equation \eqref{eq:reg_rep}.

I simulate the price path using an Euler scheme on a one-second grid. Then I sample the prices at one-minute frequency to construct the candlestick data, implying $\Delta_n = 1/390$ and $n=390$ intraday observations. This setup mimics the empirical applications involving high-frequency financial data. I consider three different values for the local window size, $k \in \{5, 10, 20\}$. 

For each configuration, I compute empirical rejection rates for the null hypothesis $H_0: \beta_t = 0$ at significance levels of $1\%, 5\%,$ and $10\%$ using the critical values in the middle block of Table \ref{tab:beta_confint_combined}. The results are summarized in Table \ref{tab:rejection_rates}. I also report the rejection rates for the return-based beta estimator for comparison.\footnote{Note that the return-based test is developed by \cite{bollerslev2024optimalspotregressions}. As discussed in Section \ref{sec:inference}, the corresponding test statistic is shown to be Student's t distributed.} All results are based on 10,000 Monte Carlo simulations.

\begin{table}[t!]
\centering

\caption{\textbf{Power of Tests Derived from Return and Candlestick Estimators (\%):} The table reports the rejection rates (in percent) for the hypothesis test of $H_0: \beta_t = 0$ at significance levels $\alpha \in \{1\%, 5\%, 10\%\}$ and for different local window sizes $k \in \{5, 10, 20\}$ based on simulated data. The left panel shows the rejection rates for the return-based beta estimator, while the right panel presents the rejection rates for the candlestick-based beta estimator.}
\label{tab:rejection_rates}
\renewcommand{\arraystretch}{1.3} % increase row height
\setlength{\tabcolsep}{14pt} % adjust horizontal spacing
% adjust vertical spacing
\smallskip
\begin{tabular}{c c c c c c c c}
\toprule
 & \multicolumn{3}{c}{\textbf{Return}} & & \multicolumn{3}{c}{\textbf{Candlestick}} \\
\cmidrule(lr){2-4} \cmidrule(lr){6-8}
$\boldsymbol{k}$ & 1\% & 5\% & 10\% & & 1\% & 5\% & 10\% \\
\midrule
$5$  & 8.84  & 28.02 & 42.11 & & 31.49 & 58.82 & 65.67 \\
$10$ & 32.53 & 59.98 & 72.66 & & 60.46 & 84.81 & 91.22 \\
$20$ & 73.99 & 90.32 & 94.80 & & 94.95 & 99.15 & 99.44 \\
\bottomrule
\end{tabular} \vspace{0.25cm}
\end{table}

Simulation results demonstrate that, across all scenarios, the candlestick-based test consistently rejects false null hypotheses more frequently than the return-based test. This gap widens as the local window size decreases. For instance, when $k=20$ and $\alpha=5\%$, the difference in rejection rates is approximately $9\%$, increasing to $25\%$ for $k=10$ and $\alpha=5\%$. These results demonstrate a substantial reduction in false negatives, highlighting the superior power of the candlestick-based procedure.

To further illustrate the power advantage of the proposed test, Figure \ref{fig:power_curve} presents the power curves for both tests as a function of the true $\beta_t$ value. In this analysis, I use the same DGP as before but set $\beta_t$ to be constant, varying from $0$ to $2$. The null hypothesis is that $\beta_t=0$, and the empirical rejection rates are computed at a $5\%$ significance level using $10,000$ Monte Carlo simulations. The local window size is fixed at $k=10$.
\begin{figure}[t!]
    \centering
    \includegraphics[width=0.95\textwidth]{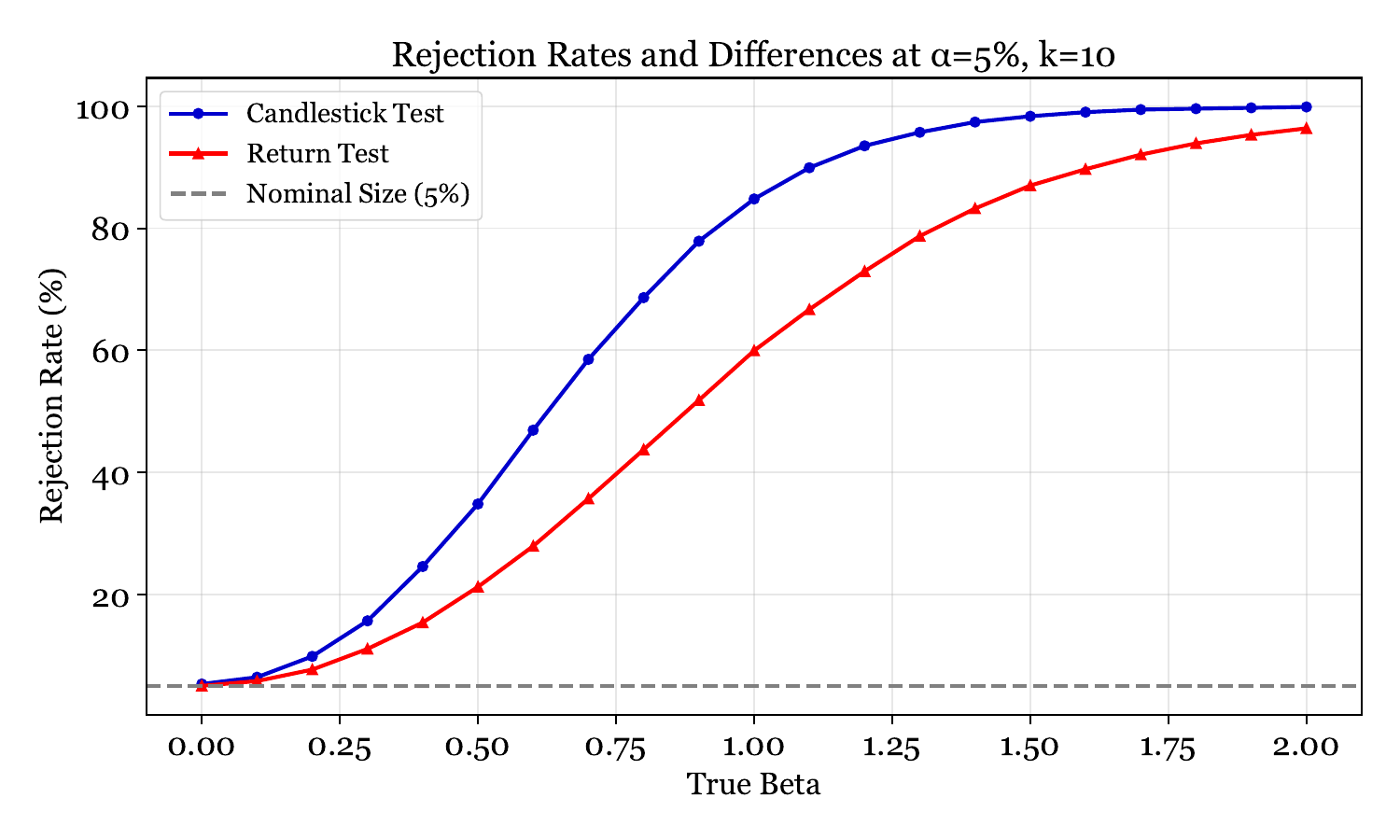}
    \caption{\textbf{Power Curves for} $H_0: \beta_t = 0\ $: The figure plots the empirical power of the return-based and candlestick-based tests as a function of the true $\beta_t$ value (assumed to be constant), with $k=10$ and a significance level of $5\%$. The power is computed using 10,000 Monte Carlo simulations.}
    \label{fig:power_curve}
\end{figure}

Figure \ref{fig:power_curve} shows the rejection rates for both return-based and candlestick-based tests. The x-axis shows how far $\beta_t$ is from the null value of $0$. Looking at the leftmost point where $\beta_t$ is set to zero, both tests appear to correctly maintain the size at approximately $5\%$. At the right-end where $\beta_t$ is far from the null, unsurprisingly, both tests achieve a power close to $100\%$. As $\beta_t$ deviates from zero, the power of both tests increases, with the candlestick-based test exhibiting a notably steeper increase. For example, when $\beta_t$ is around $1$, the candlestick-based test achieves a power of approximately $85\%$, while the return-based test is around $60\%$, indicating a $25\%$ difference.

So far, this analysis has focused on the case of a zero null hypothesis. I also evaluate the power of the tests for non-zero null hypotheses. Particularly, I consider $H_0: \beta_t = 1$. For this case, the same simulated price data is employed as before, but the critical values are now based on rightmost block of Table \ref{tab:beta_confint_combined}. Figure \ref{fig:power_curve_suptest} exhibits the power curves for both tests under this non-zero null hypothesis. A similar pattern emerges here: the candlestick-based test consistently outperforms the return-based test across almost all values of $\beta_t$, with the power gap increases up to $25\%$ as $\beta_t$ moving away from the null value of $1$. Not surprisingly, at the null, the candlestick-based test slightly underrejects, around $4\%$, while the return-based test maintains the correct size. 

Overall, these results demonstrate that candlestick-based inference substantially improves power for spot betas, particularly at moderate deviations from the null. Furthermore, I show that the advantage extends beyond the zero-null case to more general settings in which sup test approach is implemented, underscoring the practical value of the proposed test for inference on spot betas.

\begin{figure}[t!]
    \centering
    \includegraphics[width=0.95\textwidth]{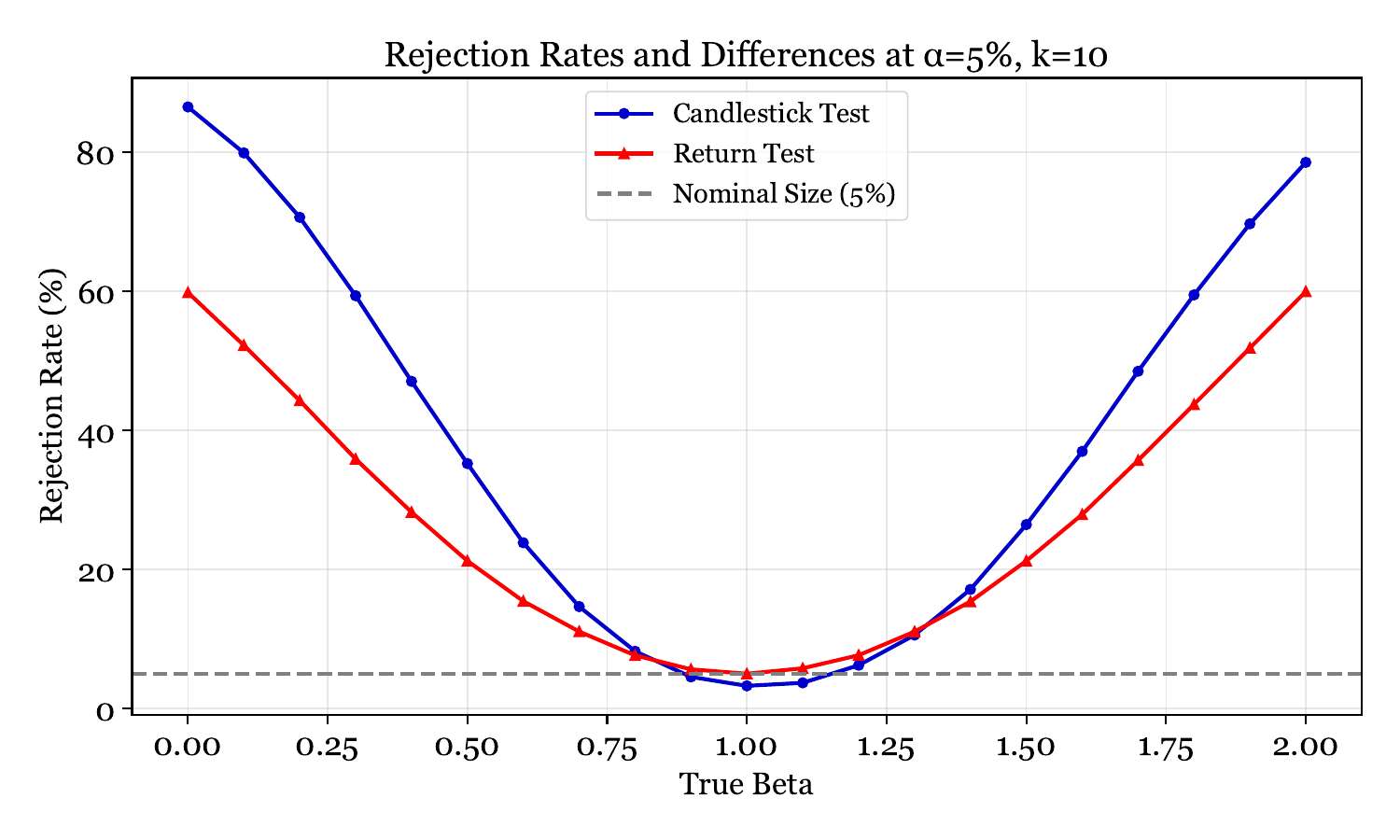}
    \caption{\textbf{Power Curves for} $H_0: \beta_t = 1\ $: The figure plots the empirical power of the return-based and candlestick-based tests as a function of the true $\beta_t$ value (assumed to be constant), with $k=10$ and a significance level of $5\%$. The power is computed using 10,000 Monte Carlo simulations.}
    \label{fig:power_curve_suptest}
\end{figure}

\section{Empirical Application: Is Bitcoin Market Neutral?} \label{sec:empirics}
To demonstrate the practical value of the proposed estimator and inference procedure, I analyze Bitcoin's market exposure, i.e., its market beta. Given the growing role of crypto assets in institutional and retail portfolios, this empirical question can have important implications for risk management and portfolio selection. Specifically, recent years have seen several developments that facilitated investment in crypto assets. In 2017, the Chicago Mercantile Exchange (CME) launched Bitcoin futures contracts, followed by the introduction of Bitcoin options in 2020. More recently, in January 2024, the U.S. Securities and Exchange Commission (SEC) approved the first Bitcoin exchange-traded fund (ETF).

Meanwhile, the risk characteristics of crypto assets remain a subject of debate among academics and practitioners. Crypto advocates often describe these assets as ``digital gold'', suggesting potential hedging benefits against aggregate market risk. In a similar vein, \cite{liu2021risks} find limited evidence of systematic exposure of crypto assets to traditional factors including the equity market portfolio. Motivated by these developments, I estimate the spot beta and then test the null of market neutrality, i.e., zero beta, using the new candlestick-based framework proposed in Section \ref{sec:inference}, offering a more granular perspective to the same empirical question.

%\subsection{Data}
For this empirical analysis, I collect $1$-minute candlestick observations from two prominent ETFs: SPY and IBIT. While the former is a well-known ETF that tracks the S\&P 500 index and commonly used in empirical studies, the latter is a newly launched (as of January 2024) iShares Bitcoin Trust ETF designed to track the performance of Bitcoin.\footnote{Cryptocurrency ETFs are designed to provide investors with exposure to the price movements in crypto markets. The IBIT has been the most traded one since its launch and its net asset value exceeds $\mathdollar$70 billion as of 2025.} The data is sourced from NYSE Trade and Quote database through Wharton Research Data Services(WRDS). The sample spans the entire year 2024, covering $250$ trading days and standard trading hours from $9:30$ to $16:00$.\footnote{One advantage of pairing SPY with IBIT is that both ETFs trade on regulated exchanges (NYSE Arca and NASDAQ) and thus share the same market microstructure like the same trading hours. This means that the co-movement between them resulted from genuine economic dynamics, not from differences in how or when each asset is traded. That makes the pair a clean setting for my empirical analysis.}

%\subsection{Is Bitcoin Market Neutral?}
I estimate spot beta using non-overlapping local windows of $k=10$ one-minute candlesticks, yielding $39$ spot estimates per trading day. Then, I test the null hypothesis of zero beta $H_0:\beta_t=0$, or market neutrality, for each estimate using the test statistic developed in Section \ref{sec:inference} with a significance level of $5\%$. 

\begin{figure}[t!]
    \centering
    \includegraphics[width=0.9\textwidth]{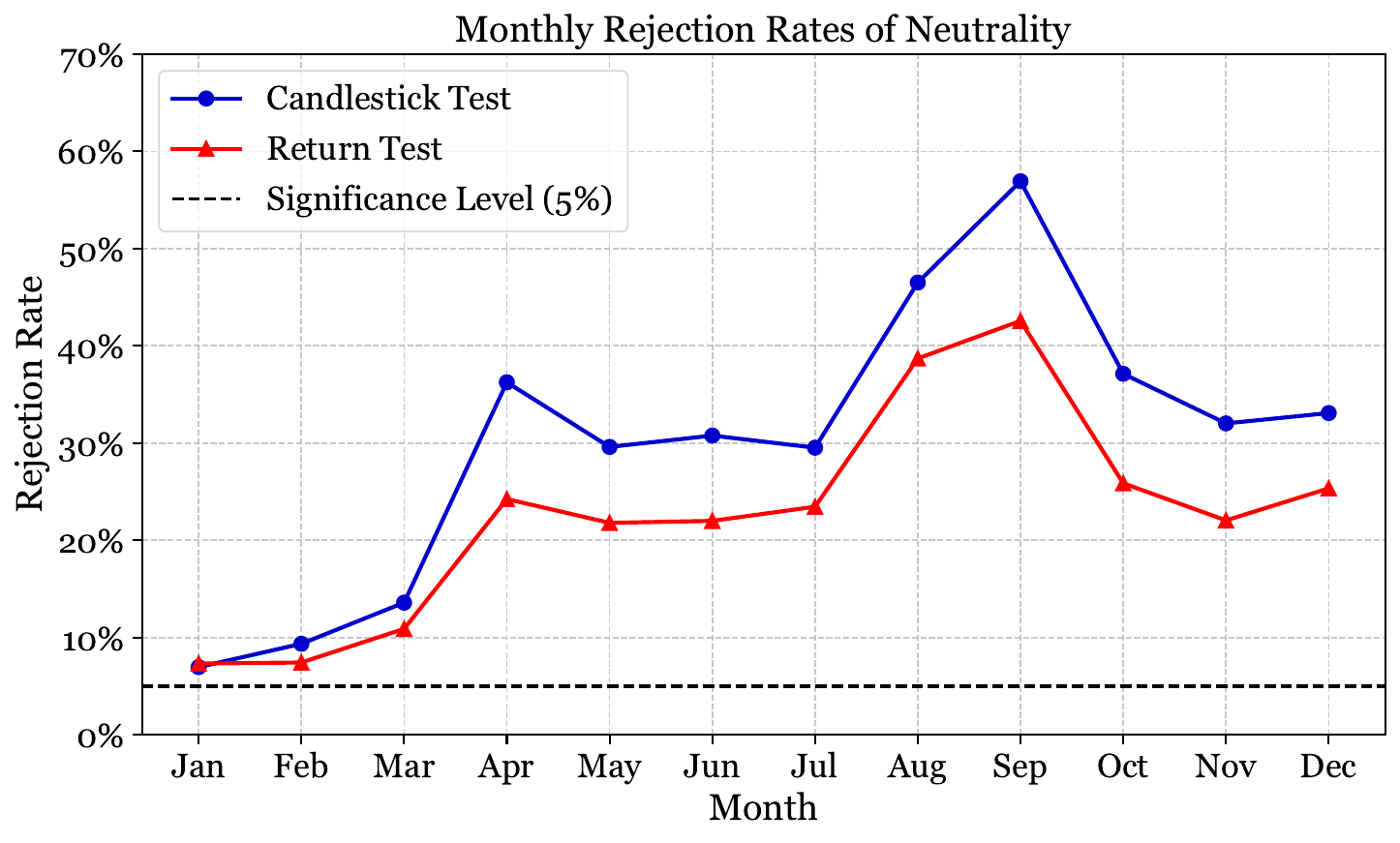}
    \caption{\textbf{Monthly Rejection Rates of the Null Hypothesis of Zero Beta for IBIT.} The figure shows the monthly rejection rates of the null hypothesis of zero beta for IBIT using candlestick-based inference framework. The sample covers the entire 2024 year, a total of $250$ days, and usual trading hours from $9:30$ to $16:00$.}
    \label{fig:ibit_beta_rejection}
\end{figure}

I start by assessing the monthly rejection rates (in terms of percentage) of the null hypothesis of market neutrality which is defined as the proportion of intraday spot beta estimates that are significantly different from zero at the 5\% level. The results are presented in Figure \ref{fig:ibit_beta_rejection}. The figure shows that the rejection rates are around $10\%$ in the first two months of 2024, later increasing to $40\%$ in mid-2024 and ending the year with a similar rate. Notably, the rejection rates are more pronounced in August and September, reaching around $60\%$. 

Interestingly, these months also coincide with a number of crucial economic events and heightened market volatility. For instance, in August 2024, Bank of Japan announced a rate hike and this was followed by global equity sell-offs. Moreover, the first week of September was marked by a series of weak production and labor market data releases, which raised concerns about a potential economic slowdown. Specifically, during that period, the ISM manufacturing index, the ADP jobs report and the non-farm payrolls came in below expectations and consequently increased market volatility.

It may be instructive to look at these days in more detail. Figure \ref{fig:ibit_spot_beta_sept36} presents the spot beta estimates and corresponding confidence intervals for September 3–6. This figure reveals that the null hypothesis of zero beta is rejected in a substantial portion of the intraday intervals, approximately $65\%$ of the time. For instance, looking at the bottom right panel, the null is rejected in $27$ out of $39$ intervals on September 6. On that day, the NFP report was released at 8:30 AM, prior to market open, and spot beta estimates were already significant and around $1.5$ at the opening. Similar patterns are observed on the other days of that week. 

As a result, this empirical analysis provides evidence that Bitcoin's market exposure may not be negligible, and it can show significant positive beta with respect to the market portfolio, particularly during periods of heightened market volatility, precisely the periods when such instruments are supposed to be most valuable in terms of risk management purposes.

\begin{figure}[t!]
    \centering
    \includegraphics[width=0.99\textwidth]{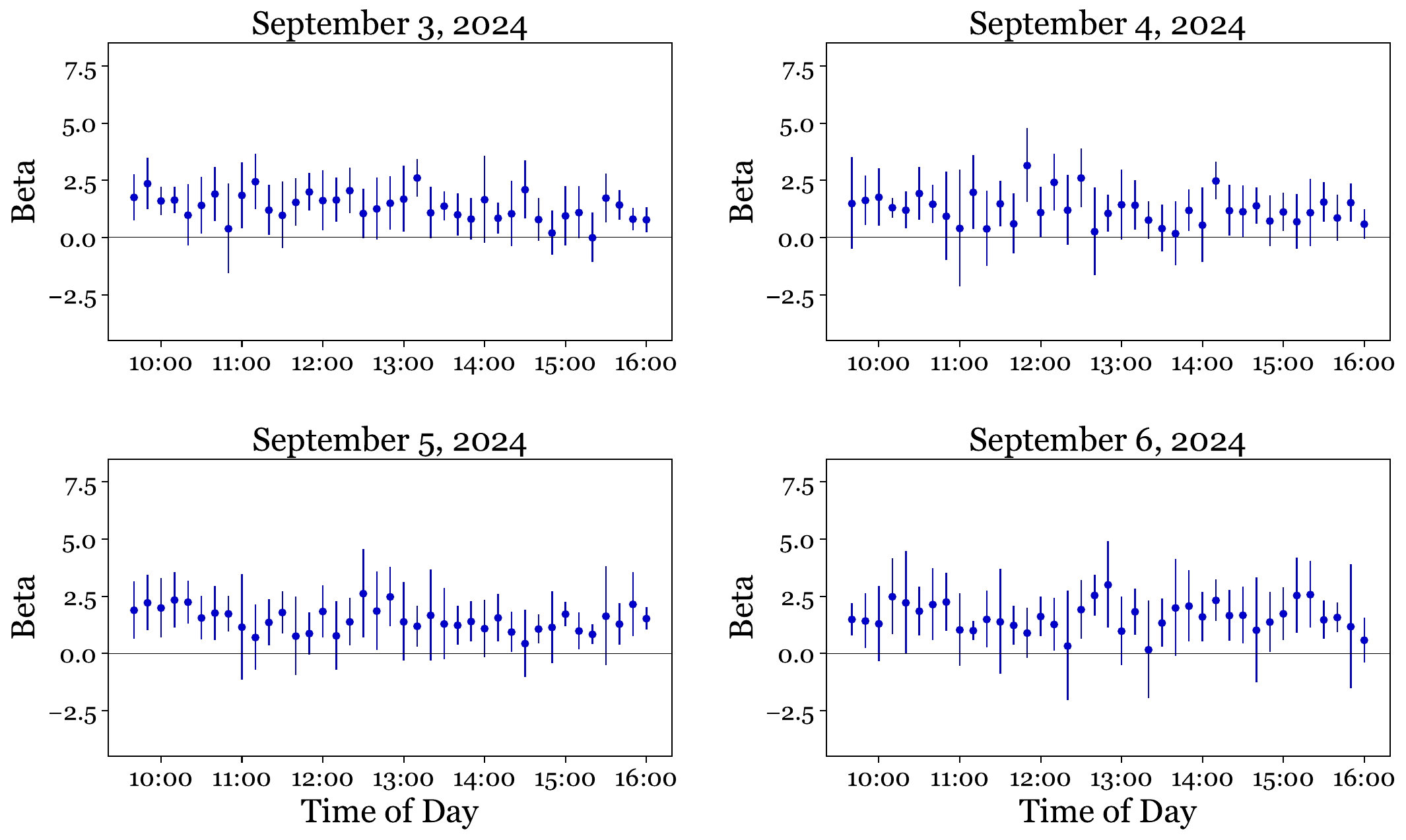}
    \caption{\textbf{Spot Beta of IBIT on September 3-6, 2024.} The figure shows the spot beta estimates of IBIT with respect to SPY using $10$ $1-min$ candlestick observations in $10-min$ frequency. The vertical lines indicate the 95\% confidence intervals.}
    \label{fig:ibit_spot_beta_sept36}
\end{figure}

\section{Conclusion} \label{sec:conclusion}
This paper develops a new framework for estimating and conducting inference on spot regressions, with a focus on spot betas, using high-frequency candlestick data, which contains high and low prices in addition to the open and close prices for each interval. The estimator is constructed by minimizing a quadratic risk function under fixed-$k$ asymptotics, where $k$ is the number of intraday observations used for local estimation. This approach yields a simple and implementable estimator that effectively combines information from all four prices. Furthermore, I develop a hypothesis testing procedure for spot betas based on the proposed estimator, where the critical values are computed by simulating the limiting distribution of the test statistic under the null hypothesis. 

Monte Carlo experiments demonstrate that the candlestick-based estimator attains asymptotic risk levels close to those of an infeasible oracle benchmark, while substantially outperforming the return-based estimator. Moreover, the candlestick-based test exhibits markedly higher power (up to $25\%$ improvement) and reduces false negatives in inference. These efficiency gains are especially valuable for high-frequency event studies, where identification relies on short time windows. In an empirical study with 1-minute candlestick data of SPY (S\&P 500 ETF) and IBIT (iShares Bitcoin Trust ETF), the method reveals significant positive market exposure of the Bitcoin ETF, especially during turbulent periods, challenging popular ``digital gold'' narrative for risk management. 

%Looking ahead, it would be interesting to extend this paper's approach to estimate covariance matrix of many assets, which is crucial for portfolio optimization. This might be achieved thorugh a factor-based approach as in \cite{bollerslev2019high} in which candlestick-based beta estimates can be used to construct factor exposures. I leave this extension for future research.

Overall, this paper highlights the practical value of candlestick data for more precise and reliable estimation and inference in spot regressions, offering researchers and practitioners a readily implementable tool for analyzing time-varying risk exposures in financial markets.

%% ====== APPENDIX ======
\newpage
\appendix

\renewcommand{\thesection}{Appendix \Alph{section}}
\renewcommand{\thesubsection}{\Alph{section}.\arabic{subsection}}

%% APPENDIX A
\section{} 

% Some adjustments
\renewcommand{\thetable}{A.\arabic{table}}  
\renewcommand{\thefigure}{A.\arabic{figure}}
\renewcommand{\theequation}{A.\arabic{equation}}

\setcounter{equation}{0} 
\setcounter{assumption}{0}
\setcounter{theorem}{0}
\setcounter{proposition}{0}

\subsection{Decomposition of the Loss Function} \label{app:loss_decomp}
The remark in the Section \ref{sec:optimal_weight} presents a decomposition of the loss function into three components plus cross terms: $L(\lambda; \boldc_t) = L_{\nu} + L_{\varsigma} + L_{\beta} + \text{cross terms}$. This decomposition is crucial for understanding the estimation methodology of this paper. Therefore, I provide further details on how this decomposition is derived and its implications for the estimation procedure, in this section.

To derive the decomposition, I start by writing the estimator as Cholesky form:
\begin{equation}
    \widehat{\boldc}_{n,t}(\lambda) = \widehat{\boldsigma}_{n,t}(\lambda) \widehat{\boldsigma}_{n,t}(\lambda)^{\top} \quad \text{where} \quad \widehat{\boldsigma}_{n,t}(\lambda) = \begin{pmatrix}
        \widehat{\nu}_t^{1/2} & 0 \\
        \widehat{\beta}_t \widehat{\nu}_t^{1/2} & \widehat{\varsigma}_t^{1/2}
    \end{pmatrix}.
\end{equation}
Now, consider multiplying $\widehat{\boldsigma}_{n,t}(\lambda)$ by $\boldsigma_t^{-1}$:
\begin{equation*}
\begin{aligned}
    \boldsigma_t^{-1}\widehat{\boldsigma}_{n,t}(\lambda) &= \begin{pmatrix}
        \nu_t^{-1/2} & 0 \\
        -\beta_t \nu_t^{-1/2} & \varsigma_t^{-1/2}
    \end{pmatrix} \begin{pmatrix}
        \widehat{\nu}_t^{1/2} & 0 \\
        \widehat{\beta}_t \widehat{\nu}_t^{1/2} & \widehat{\varsigma}_t^{1/2}
    \end{pmatrix} 
    &= \begin{pmatrix}
        (\widehat{\nu}_t/\nu_t)^{1/2} & 0 \\
        (\widehat{\beta}_t - \beta_t)(\widehat{\nu}_t/\nu_t)^{1/2} & (\widehat{\varsigma}_t/\varsigma_t)^{1/2}
    \end{pmatrix}.
\end{aligned}
\end{equation*}
To simplify the notation, define the following variables:
\begin{equation*}
    A \equiv (\widehat{\nu}_t/\nu_t)^{1/2} - 1, \quad
    B \equiv (\widehat{\varsigma}_t/\varsigma_t)^{1/2} - 1 \quad \text{and} \quad
    C \equiv (\widehat{\beta}_t - \beta_t)(\widehat{\nu}_t/\nu_t)^{1/2}
\end{equation*}
With these definitions, the multiplicative estimation error can be expressed as:
\begin{equation*}
    \boldsigma_t^{-1}\widehat{\boldc}_{n,t}(\lambda)\boldsigma_t^{-1\top} = \begin{pmatrix}
        A^2 & AB \\
        AB & B^2 + C^2
    \end{pmatrix}.
\end{equation*}
Then, the loss function can be written as:
\begin{equation*}
\begin{aligned}
    L(\lambda; \boldc_t) &= \left\| \boldsigma_t^{-1}\widehat{\boldc}_{n,t}(\lambda)\boldsigma_t^{-1\top} - I \right\|^2 \\
    &= (A^2 - 1)^2 + (B^2 + C^2 - 1)^2 + 2(AB)^2 
\end{aligned}
\end{equation*}
After some algebraic manipulation, one can further express the loss function as:
\begin{equation*}
    L(\lambda; \boldc_t) = (A^2 - 1)^2 + (C^2 - 1)^2 + \frac{B^2}{A^2}(2A^4 + A^2B^2 + 2A^2(C^2-1) )
\end{equation*}
Note that
\begin{equation*}
    (A^2 - 1)^2 = \left(\frac{\widehat{\nu}_t - \nu_t}{\nu_t}\right)^2,\quad (C^2 - 1)^2 = \left(\frac{\widehat{\varsigma}_t - \varsigma_t}{\varsigma_t}\right)^2 \quad \text{and} \quad \frac{B^2}{A^2} = \left(\frac{\widehat{\beta}_t - \beta_t}{\varsigma_t/\nu_t}\right)^2.
\end{equation*}
Finally, the loss function can be decomposed as:
 \begin{equation*}
    L(\lambda; \boldc_t) = \underbrace{\left(\frac{\widehat{\nu}_t - \nu_t}{\nu_t}\right)^2}_{\equiv L_{\nu}} + \underbrace{\left(\frac{\widehat{\varsigma}_t - \varsigma_t}{\varsigma_t}\right)^2}_{\equiv L_{\varsigma}} + \underbrace{\left(\frac{\widehat{\beta}_t - \beta_t}{\varsigma_t/\nu_t}\right)^2}_{\equiv L_{\beta}} \cdot ~ \psi
\end{equation*}
where $L_{\nu}$, $L_{\varsigma}$ and $L_{\beta}$ represent the squared relative errors of the market variance, idiosyncratic variance and spot beta, respectively, and $\psi = 2A^4 + A^2B^2 + 2A^2(C^2-1)$ is a scaling factor and function of the estimation errors of all three parameters. 

This decomposition reveals that the loss function can be expressed as a sum of weighted squared errors of the three parameters of interest: the market variance $\nu_t$, the idiosyncratic variance $\varsigma_t$ and the spot beta $\beta_t$. Thus, minimizing the overall loss function $L(\lambda; \boldc_t)$ effectively balances tradeoffs between these different components.

\subsection{Proof of Proposition \ref{prop:spotcov_coupling}} \label{app:prop_1_proof}
In this section, I provide the proof of Proposition \ref{prop:spotcov_coupling} which builds on the coupling techniques developed in \cite{jacod2021volatility} and \cite{bollerslev2021fixed}. The former studies the approximation of the estimation error for the spot covariance estimator when $k$ increases with $n$, whereas the latter focuses on coupling the spot volatility estimator in a fixed-$k$ framework. As noted in the main text, I consider a fixed-$k$ setup and thus my work is in the same spirit as \cite{bollerslev2021fixed}. To ensure consistency with the existing literature, I closely follow the notation introduced in the aforementioned papers.

I rewrite the Assumption \ref{asmp:spotcov_asmp} and Proposition \ref{prop:spotcov_coupling} here for convenience:
\begin{assumption} \label{app:asmp:spotcov_asmp}
    Suppose that $\boldsymbol{X}_t$ has the form in Equation \eqref{eq:itosemi} and there exists a sequence $(T_m)_{m \geq 1}$ of stopping times increasing to infinity and the following conditions hold for each $m \geq 1$:
    \begin{itemize}
        \item[(i)] $\|\boldsymbol{b}_t\| + \|\boldsymbol{\sigma}_t\| + \|\boldsigma_t^{-1}\| \leq K_m$ for some constant $K_m$ for all $t \in [0, T_m]$;
        \item[(ii)] $\mathbb{E}\left[ \|\boldsigma_{t \wedge T_m}-\boldsigma_{s\wedge T_m} \|^2 \right] \leq K_m |t-s|$ for all $t, s \in [0, T_m]$.
    \end{itemize}
\end{assumption}   

\begin{proposition} \label{app:prop:spotcov_coupling}
    Suppose that Assumption \ref{asmp:spotcov_asmp} holds. Fix any $t \in [0,T]$. For any $k \geq 1$ and $\boldLambda$, the following holds as $\Delta_n \to 0$:

    \begin{equation} 
        \left\| \boldsymbol{\sigma}_t^{-1}\widehat{\boldsymbol{c}}_{n, t}(\lambda)\boldsymbol{\sigma}_t^{-\top} -  U_{n,t} \right\| = o_p(1)
    \end{equation}
    where $U_{n,t} = \frac{1}{k}\sum_{i\in \mathcal{I}_{n,t}} \Big\{\lambda_1 \boldzeta_{i,r} \boldzeta_{i,r}^{\top} + \lambda_2 \boldzeta_{i,a} \boldzeta_{i,a}^{\top} + \lambda_3 \boldzeta_{i,w} \boldzeta_{i,w}^{\top} \Big\}$ and, for any $i\in \mathcal{I}_{n,t}$,
    \begin{equation}
    \begin{array}{rcl}
        \boldsymbol{\zeta}_{i, r} &\equiv&  
         \frac{\boldW_{i\Delta_n} - \boldW_{(i-1)\Delta_n}}{\sqrt{\Delta_n}} \\
         \boldsymbol{\zeta}_{i, a} &\equiv& 
          \boldvarrho_{t}^{-1}\underset{\tau \in \mathcal{T}_i}{\sup \:} \boldvarrho_{t} \left(\frac{\boldW_{\tau} - \boldW_{(i-1)\Delta_n}}{\sqrt{\Delta_n}}\right) + \boldvarrho_{t}^{-1}\underset{\tau \in \mathcal{T}_i}{\inf \:} \boldvarrho_{t}   \left(\frac{\boldW_{\tau} - \boldW_{(i-1)\Delta_n}}{\sqrt{\Delta_n}}\right) - \left(\frac{\boldW_{i\Delta_n} - \boldW_{(i-1)\Delta_n}}{\sqrt{\Delta_n}}\right) \\
         \boldsymbol{\zeta}_{i, w} &\equiv& 
          \boldvarrho_{t}^{-1}\underset{\tau \in \mathcal{T}_i}{\sup \:} \boldvarrho_{t}   \left(\frac{\boldW_{\tau} - \boldW_{(i-1)\Delta_n}}{\sqrt{\Delta_n}}\right) - \boldvarrho_{t}^{-1}\underset{\tau \in \mathcal{T}_i}{\inf \:} \boldvarrho_{t}   \left(\frac{\boldW_{\tau} - \boldW_{(i-1)\Delta_n}}{\sqrt{\Delta_n}}\right) \\
    \end{array}
    \end{equation}
    with $\boldvarrho_t$ being the square root of spot correlation matrix $\boldsymbol{\rho}_t$, i.e., $\boldsymbol{\rho}_t = \boldvarrho_t \boldvarrho_t^\top$. In explicit terms, $\boldsymbol{\rho}_t = \operatorname{diag}(\boldc_t)^{-\frac{1}{2}}\boldc_t \operatorname{diag}(\boldc_t)^{-\frac{1}{2}}$ and $\boldvarrho_t = \operatorname{diag}(\boldc_t)^{-\frac{1}{2}} \boldsigma_t$ where $\operatorname{diag}(\boldc_t)$ is a diagonal matrix with the same diagonal elements as $\boldc_t$.
\end{proposition}

\begin{proof}
    Fix $k \geq 1$ and $\lambda$. Let $K$ denote a generic positive constant. As is common in the spot estimation literature, one can strengthen Assumption \ref{asmp:spotcov_asmp} by assuming the conditions hold with $T_m = \infty$, which can be justified by a standard localization argument (see \cite{jacod2012discretization} for details).

    I begin by writing out the explicit expressions for the candlestick returns:
    \begin{equation} \label{app:eq:candlestick_returns}
        \begin{array}{rcl}
            \boldr_i & \equiv & \Delta_n^{-\frac{1}{2}} \left(\int_{(i-1)\Delta_n}^{i\Delta_n} \boldsymbol{b}_s ds + \int_{(i-1)\Delta_n}^{i\Delta_n} \boldsymbol{\sigma}_s d \boldW_s \right) \vspace{1em}\\
            \boldh_i & \equiv & \Delta_n^{-\frac{1}{2}} \left( \sup_{t \in \mathcal{T}_{n,i}} \left( \int_{(i-1)\Delta_n}^{t} \boldsymbol{b}_s ds + \int_{(i-1)\Delta_n}^{t} \boldsymbol{\sigma}_s d \boldW_s \right) \right) \vspace{1em}\\
            \boldl_i & \equiv & \Delta_n^{-\frac{1}{2}} \left( \inf_{t\in \mathcal{T}_{n,i}} \left( \int_{(i-1)\Delta_n}^{t} \boldsymbol{b}_s ds + \int_{(i-1)\Delta_n}^{t} \boldsymbol{\sigma}_s d \boldW_s \right) \right)
        \end{array}
    \end{equation}
    I also introduce the following definitions:
    \begin{equation} \label{app:eq:coupling_returns}
        \begin{array}{rcl}
            \boldr_i' & \equiv & \boldsigma_{(i-1)\Delta_n}\left(\frac{\boldW_{i\Delta_n}-\boldW_{(i-1)\Delta_n}}{\sqrt{\Delta_n}}\right)\vspace{1em}\\
            \boldh_i' & \equiv & \sup_{t\in \mathcal{T}_{n, i}} \boldsigma_{(i-1)\Delta_n}\left(\frac{\boldW_{t}-\boldW_{(i-1)\Delta_n}}{\sqrt{\Delta_n}}\right)\vspace{1em}\\
            \boldl_i' & \equiv & \inf_{t\in \mathcal{T}_{n, i}} \boldsigma_{(i-1)\Delta_n}\left(\frac{\boldW_{t}-\boldW_{(i-1)\Delta_n}}{\sqrt{\Delta_n}}\right)
        \end{array}
    \end{equation}
    which serve as the coupling variables for the candlestick returns. Similar variables can be defined for the range and asymmetry variables:
    \begin{equation} 
        \begin{array}{rcl}
            \bolda_i' &\equiv& \boldh_i' + \boldl_i' - \boldr_i' \\
            \boldw_i' &\equiv& \boldh_i' - \boldl_i'
        \end{array}
    \end{equation}

    The proof consists of two steps. The first step controls how well the coupling returns in Equation \eqref{app:eq:coupling_returns} approximate the candlestick returns in Equation \eqref{app:eq:candlestick_returns}. The second step combines these results with the continuous mapping theorem to reach the desired conclusion. Before proceeding to the first step, I derive useful intermediate results.

    By Assumption \ref{asmp:spotcov_asmp}, it is easy to see that: 
    \begin{equation} \label{eq:app:drift_bound}
        \left \| \int_{(i-1)\Delta_n}^{i\Delta_n} \boldsymbol{b}_s ds \right \| \leq \int_{(i-1)\Delta_n}^{i\Delta_n} \| \boldsymbol{b}_s \| ds = O_p(\Delta_n).
    \end{equation}
    Moreover, the Burkholder-Davis-Gundy inequality and Assumption \ref{asmp:spotcov_asmp} imply that:
    \begin{equation}
        \begin{array}{rcl}
            \mathbb{E}\left[ \left\| \int_{(i-1)\Delta_n}^{i\Delta_n}(\boldsymbol{\sigma}_s - \boldsigma_{(i-1)\Delta_n}) d\boldW_s \right\|^2 \right] &\leq& K \Delta_n \mathbb{E}\left[ \int_{(i-1)\Delta_n}^{i\Delta_n} \|\boldsymbol{\sigma}_s - \boldsigma_{(i-1)\Delta_n}\|^2 ds \right] \\
            &\leq& K \Delta_n^2.            
        \end{array}
    \end{equation}
    Further, we can deduce that:
    \begin{equation} \label{eq:app:volatility_path_bound}
        \sup_{t \in \mathcal{T}_{n,i}} \left\| \int_{(i-1)\Delta_n}^{t}(\boldsymbol{\sigma}_s - \boldsigma_{(i-1)\Delta_n}) d\boldW_s \right\| = O_p(\Delta_n).
    \end{equation}

    \textit{Step 1:} I now establish approximation results for $\boldr_i$, $\boldh_i$ and $\boldl_i$ separately. I start with the return:
    \begin{equation}
        \begin{array}{rcl}
            \| \boldr_i - \boldr_i' \| &\leq & \left \| \Delta_n^{-\frac{1}{2}}\int_{(i-1)\Delta_n}^{i\Delta_n} \boldsymbol{b}_s ds \right \|+ \left \| \Delta_n^{-\frac{1}{2}} \int_{(i-1)\Delta_n}^{i\Delta_n} (\boldsymbol{\sigma}_s - \boldsigma_{(i-1)\Delta_n}) d\boldW_s \right \| \vspace{1em} \\
            &=& O_p(\Delta_n^{\frac{1}{2}}).
        \end{array}
    \end{equation}
    where the first line directly follows from the triangle inequality and the second line uses above intermediate results in Equation \eqref{eq:app:drift_bound} and \eqref{eq:app:volatility_path_bound}. For the high return, we have:
    \begin{equation}
        \begin{array}{rcl}
            \left\| \boldh_i - \boldh_i' \right\| 
            &\equiv& \Delta_n^{-\frac{1}{2}} \Big\| 
            \sup_{t \in \mathcal{T}_{n,i}} \Big( \int_{(i-1)\Delta_n}^{t} \boldsymbol{b}_s ds + \int_{(i-1)\Delta_n}^{t} \boldsymbol{\sigma}_s d\boldW_s \Big) \\
            && \qquad - \sup_{t \in \mathcal{T}_{n,i}} \boldsigma_{(i-1)\Delta_n} (\boldW_t - \boldW_{(i-1)\Delta_n}) \Big\| \vspace{1em}\\
            &\leq& \Delta_n^{-\frac{1}{2}} \sup_{t\in \mathcal{T}_{n,i}} \left\| \int_{(i-1)\Delta_n}^{t} \boldsymbol{b}_s ds + \int_{(i-1)\Delta_n}^{t} (\boldsymbol{\sigma}_s - \boldsigma_{(i-1)\Delta_n}) d\boldW_s \right\| \vspace{1em} \\
            &\leq& \Delta_n^{-\frac{1}{2}} \int_{(i-1)\Delta_n}^{i\Delta_n} \|\boldsymbol{b}_s\| ds + \sup_{t\in \mathcal{T}_{n,i}} \left\| \int_{(i-1)\Delta_n}^{t} (\boldsymbol{\sigma}_s - \boldsigma_{(i-1)\Delta_n}) d\boldW_s \right\| \vspace{1em}\\
            &=& O_p(\Delta_n^{\frac{1}{2}}).
        \end{array}
    \end{equation}
    The first two lines are obviously implications of $\sup$ definition. Similarly, the last line follows from Equation \ref{eq:app:drift_bound} and \ref{eq:app:volatility_path_bound}. Finally, one can deduce the same inequality for the low return: 
    \begin{equation}
        \left \| \boldl_i - \boldl_i' \right \| = O_p(\Delta_n^{\frac{1}{2}}),
    \end{equation}
    and also for the asymmetry and range variables:
    \begin{equation}
        \begin{array}{rcl}
            \| \bolda_i - \bolda_i' \| &\leq& \| \boldh_i - \boldh_i' \| + \| \boldl_i - \boldl_i' \| + \| \boldr_i - \boldr_i' \| = O_p(\Delta_n^{\frac{1}{2}}) \vspace{1em}\\
            \| \boldw_i - \boldw_i' \| &\leq& \| \boldh_i - \boldh_i' \| + \| \boldl_i - \boldl_i' \| = O_p(\Delta_n^{\frac{1}{2}}).
        \end{array}
    \end{equation}

    Furthermore, I claim that 
     \begin{equation}
        \begin{array}{rcl}
            \Big \| \boldsigma_t^{-1} \boldr_i^{'} - \boldzeta_{i,r} \Big\| & = & O_p(\Delta_n^{1/2}) \vspace{1em}\\
            \Big \| \boldsigma_t^{-1} \bolda_i^{'} - \boldzeta_{i,a} \Big \| & = & O_p(\Delta_n^{1/2}) \vspace{1em}\\
            \Big \| \boldsigma_t^{-1} \boldw_i^{'} - \boldzeta_{i,w} \Big \| & = & O_p(\Delta_n^{1/2}).
        \end{array}
    \end{equation}
    The first line directly follows from Assumption \ref{asmp:spotcov_asmp}. Note that $|i\Delta_n - t| \to 0$ for any $i \in \mathcal{I}_{n,t}$ as $\Delta_n \to 0$ and this implies $\| \boldsigma_t - \boldsigma_{(i-1)\Delta_n}\| = O_p(\Delta_n^{1/2})$. For notational convenience, I only consider the third line and the same steps can be adopted for the second line as well. Specifically, one can write:
    \begin{equation}
        \begin{array}{rcl}
            \| \boldsigma^{-1}_t \boldw_i^{'} - \boldzeta_{i, w} \| & \leq & \| \boldsigma_t^{-1} \| \cdot \| \boldw_i^{'} - \boldsigma_t\boldzeta_{i, w}\| \vspace{1em}\\
            & = & \| \boldsigma_t^{-1} \| \cdot \Big \| \boldw_i^{'} - \operatorname{diag}(\boldc_t)^{1/2}\sup_{\tau, s\in \mathcal{T}_{n,i}} \boldvarrho_{t} \left(\frac{\boldW_{\tau} - \boldW_s}{\sqrt{\Delta_n}}\right) \Big \| \vspace{1em}\\
            & = & \| \boldsigma_t^{-1} \| \cdot \Big \| \boldw_i^{'} - \sup_{\tau, s\in \mathcal{T}_{n,i}} \boldsigma_{t} \left(\frac{\boldW_{\tau} - \boldW_s}{\sqrt{\Delta_n}}\right) \Big \| \vspace{1em}\\
            & = & \| \boldsigma_t^{-1} \| \cdot \Big \|\sup_{\tau, s\in \mathcal{T}_{n,i}} (\boldsigma_{t}-\boldsigma_{(i-1)\Delta_n}) \left(\frac{\boldW_{\tau} - \boldW_s}{\sqrt{\Delta_n}}\right) \Big \| \vspace{1em}\\
            & = & O_p(\Delta_n^{1/2})\\   
        \end{array}
    \end{equation}
    where the first line follows from sub-multiplicative property of matrix norm, the second and third lines use the definition of $\boldzeta_{i,w}$ and $\boldvarrho_t$, the fourth line is a direct implication of $\sup$ definition and the last line uses Assumption \ref{asmp:spotcov_asmp} and properties of Brownian motion.

    \textit{Step 2:} Rewrite the main statement of the proposition as:
    \begin{align}
        \left\| \boldsymbol{\sigma}_t^{-1}\widehat{\boldsymbol{c}}_{n, t}(\lambda)\boldsymbol{\sigma}_t^{-\top} -  U_{n,t} \right\| 
        &= \Bigg\| 
            \frac{1}{k} \sum_{i \in \mathcal{I}_{n,t}} \left\{ \lambda_1 (\boldsymbol{\sigma}_t^{-1}\boldsymbol{r}_i)(\boldsymbol{\sigma}_t^{-1}\boldsymbol{r}_i)^{\top} - \boldsymbol{\zeta}_{i,r}\boldsymbol{\zeta}_{i,r}^{\top} \right\} \nonumber \\
        &\quad + \frac{1}{k} \sum_{i \in \mathcal{I}_{n,t}} \left\{ \lambda_2 (\boldsymbol{\sigma}_t^{-1}\boldsymbol{a}_i)(\boldsymbol{\sigma}_t^{-1}\boldsymbol{a}_i)^{\top} - \boldsymbol{\zeta}_{i,a}\boldsymbol{\zeta}_{i,a}^{\top} \right\} \nonumber \\
        &\quad + \frac{1}{k} \sum_{i \in \mathcal{I}_{n,t}} \left\{ \lambda_3 (\boldsymbol{\sigma}_t^{-1}\boldsymbol{w}_i)(\boldsymbol{\sigma}_t^{-1}\boldsymbol{w}_i)^{\top} - \boldsymbol{\zeta}_{i,w}\boldsymbol{\zeta}_{i,w}^{\top} \right\}
        \Bigg\|
    \end{align}
    Using the results from Step 1, it follows that the terms in curly brackets are $O_p(\Delta_n^{1/2})$ for any $i \in \mathcal{I}_{n,t}$ and fixed $\lambda = (\lambda_1, \lambda_2, \lambda_3)^{\top}$. Therefore, the entire expression is $O_p(\Delta_n^{1/2}) = o_p(1)$. This completes the proof.

\end{proof}

\subsection{Proof of Proposition \ref{prop:test_coupling}} \label{app:prop_2_proof}
Now, I provide the proof of Proposition \ref{prop:test_coupling}, which establishes the coupling result for the test statistic defined as:
$$
    \widehat{T}_n = \frac{\sqrt{k-1}\left (\widehat{\beta}_{n,t} - \beta_t\right )}{\sqrt{\widehat{\varsigma}_{n,t} / \widehat{\nu}_{n,t}}}.
$$
The proof is based on the algebraic manipulations of the previous proposition. For convenience, I restate the proposition here:
\begin{proposition} \label{app:prop:test_coupling}
    Under the conditions of Proposition 1, for any fixed $k \geq 2$ and $\lambda$, the following holds as $\Delta_n \to 0$:
    $$
        |\widehat{T}_n - \widetilde{T}_n| = o_p(1) \quad \text{where} \quad \widetilde{T}_n \equiv \frac{\sqrt{k-1}[U_{n, t}^{-1}]_{12}}{\sqrt{[U_{n, t}^{-1}]_{11}[U_{n, t}^{-1}]_{22}-[U_{n, t}^{-1}]_{12}^2}}.
    $$
\end{proposition}

\begin{proof}
    Fix $k \geq 2$ and $\lambda$. By Proposition \ref{prop:spotcov_coupling}, we have:
    \begin{equation}
        \left\| \boldsymbol{\sigma}_t^{-1}\widehat{\boldsymbol{c}}_{n, t}(\lambda)\boldsymbol{\sigma}_t^{-\top} -  U_{n,t} \right\| = o_p(1).
    \end{equation}
    Note that both $\boldsymbol{\sigma}_t^{-1}\widehat{\boldsymbol{c}}_{n, t}(\lambda)\boldsymbol{\sigma}_t^{-\top}$ and $U_{n,t}$ are positive definite matrices. Therefore, by the continuity of matrix inversion operator on the set of positive definite matrices, we have:
    \begin{equation}
        \left\| (\boldsymbol{\sigma}_t^{-1}\widehat{\boldsymbol{c}}_{n, t}(\lambda)\boldsymbol{\sigma}_t^{-1\top})^{-1} -  U_{n,t}^{-1} \right\| = o_p(1).
    \end{equation}
    This implies that each element of the matrix $(\boldsymbol{\sigma}_t^{-1}\widehat{\boldsymbol{c}}_{n, t}(\lambda)\boldsymbol{\sigma}_t^{-\top})^{-1}$ converges to the corresponding element of $U_{n,t}^{-1}$ in probability. These relations can be written in explicit forms as follows:
    \begin{align*}
        \left | \left(\frac{\nu_t}{\widehat{\nu}_{n,t}} + \frac{\nu_t(\beta_t - \widehat{\beta}_{n,t})^2}{\widehat{\varsigma}_{n,t}} \right) - [U_{n,t}^{-1}]_{11} \right | &= o_p(1) \\
        \left | \left( \frac{\nu_t^{1/2}\varsigma_t^{1/2}(\beta_t - \widehat{\beta}_{n,t})}{\widehat{\varsigma}_{n,t}} \right) - [U_{n,t}^{-1}]_{12} \right | &= o_p(1) \\
        \left | \frac{\varsigma_t}{\widehat{\varsigma}_{n,t}} - [U_{n,t}^{-1}]_{22} \right | &= o_p(1).
    \end{align*}
    Using the second and third lines, one can write:
    \begin{equation}
        \left | \frac{(\widehat{\beta}_{n,t} - \beta_t)}{\sqrt{{\varsigma}_{n,t}/{\nu}_{n,t}}} - \frac{[U_{n,t}^{-1}]_{12}}{[U_{n,t}^{-1}]_{22}} \right | = o_p(1).
    \end{equation} 
    Moreover, from all three lines, one can deduce that:
    \begin{equation}
        \left | \frac{\nu_t}{\widehat{\nu}_{n,t}} - \left([U_{n,t}^{-1}]_{11} - \frac{[U_{n,t}^{-1}]_{12}^2}{[U_{n,t}^{-1}]_{22}} \right) \right | = o_p(1).
    \end{equation}
    Finally, combining the above equations yields:
    \begin{equation}
        \left | \frac{\sqrt{k-1}(\widehat{\beta}_{n,t} - \beta_t)}{\sqrt{\widehat{\varsigma}_{n,t}/\widehat{\nu}_{n,t}}} - \frac{\sqrt{k-1}[U_{n,t}^{-1}]_{12}}{\sqrt{[U_{n,t}^{-1}]_{11}[U_{n,t}^{-1}]_{22}-[U_{n,t}^{-1}]_{12}^2}} \right | = o_p(1).
    \end{equation}
    as required.

    % By the definitions of $\boldc_t$ and $\boldsigma_t$, one can write the $(\boldsymbol{\sigma}_t^{-1}\widehat{\boldsymbol{c}}_{n, t}(\lambda)\boldsymbol{\sigma}_t^{-\top})^{-1}$ as:
    % \begin{equation}
    %         (\boldsymbol{\sigma}_t^{-1}\widehat{\boldsymbol{c}}_{n, t}(\lambda)\boldsymbol{\sigma}_t^{-1\top})^{-1} 
    %          =  \begin{pmatrix}
    %             \nu_{t}^{1/2} & \beta_t \nu_{t}^{1/2} \\
    %             0 & \varsigma_{t}^{1/2}
    %         \end{pmatrix} \vspace{1em}
    %         \begin{pmatrix}
    %             \widehat{\nu}_t & \widehat{\beta}_t \widehat{\nu_t} \\
    %             \widehat{\beta}_t \widehat{\nu_t} & \widehat{\beta}_t^2 \widehat{\nu_t} + \widehat{\varsigma}_t
    %         \end{pmatrix}^{-1}
    %         \begin{pmatrix}
    %             \nu_{t}^{1/2} & 0 \\
    %             \beta_t \nu_{t}^{1/2} & \varsigma_{t}^{1/2}
    %         \end{pmatrix} 
    % \end{equation}
    % Then, the first row and first column element of this matrix is given by:

\end{proof}

%% ====== REFERENCES ======
\newpage
\bibliographystyle{apalike}
\bibliography{references.bib}

\end{document}